\documentclass[12pt]{article}

\usepackage{cite}
\usepackage{amsmath,amssymb,amsfonts}
\usepackage{amsthm}
\usepackage{graphics} % for pdf, bitmapped graphics files
\usepackage{epsfig}
\usepackage{epstopdf}
\usepackage{graphicx}
\usepackage{subfigure}
\usepackage{textcomp}
\usepackage{booktabs}
\usepackage{multirow}
\usepackage{caption}
\newtheorem{example}{Example}
\newtheorem{definition}{Definition}
\newtheorem{theorem}{Theorem}
\newtheorem{proposition}{Proposition}
\newtheorem{assumption}{Assumption}
\makeatletter
\newif\if@restonecol
\makeatother

\usepackage[linesnumbered,ruled,vlined]{algorithm2e}%[ruled,vlined]{
\usepackage{algpseudocode}
\usepackage{amsmath}

\usepackage{url}
\usepackage{hyperref}
\usepackage{authblk}

\newcommand{\tabincell}[2]{\begin{tabular}{@{}#1@{}}#2\end{tabular}}

\begin{document}

\title{\LARGE \bf {State estimation of timed automata under partial observation}}

% Add authors and affiliations
\author[1,2]{Chao Gao}
\author[3]{Dimitri Lefebvre}
\author[2]{Carla Seatzu}
\author[1,4]{Zhiwu Li}
\author[2]{Alessandro Giua}

% ORCID IDs can be added (optional)
% \author[1]{\href{https://orcid.org/0000-0001-2345-6789}{Author One ORCID}}
% \author[2]{\href{https://orcid.org/0000-0002-3456-7890}{Author Two ORCID}}

% Add affiliations
\affil[1]{\small{School of Electro-Mechanical Engineering, Xidian University, Xi'an 710071, China}}
\affil[2]{\small{DIEE, University of Cagliari, Cagliari 09124, Italy}}
\affil[3]{\small{GREAH Laboratory, Université Le Havre Normandie, Le Havre 76600, France}}
\affil[4]{\small{Institute of Systems Engineering, Macau University of Science and Technology, Macau  999078, China}}
\date{}
% \author{Chao Gao, \emph{Graduate Student Member, IEEE}, Dimitri Lefebvre, \emph{Senior Member, IEEE},\\ Carla Seatzu, \emph{Senior Member, IEEE}, Zhiwu Li, \emph{Fellow, IEEE}, and Alessandro Giua, \emph{Fellow, IEEE} 
% % <-this % stops a space
% %%\thanks{This work is partially supported by Project RASR05861 MOSIMA financed by Region Sardinia, FSC 2014-2020, annuity 2017, Subject Area 3, Action line 3.1.}% <-this % stops a space
% \thanks{This work is partially supported by the National Key R\&D Program under Grant 2018YFB700104, National Natural Science Foundation of China under Grant 61873342(Corresponding author: Zhiwu Li).
% }
% \thanks{Chao Gao is with the School of Electro-Mechanical Engineering,
% Xidian University, Xi'an 710071, China; DIEE, University of Cagliari, Cagliari 09124, Italy, E-mail:
%         {\tt\small gaochao@stu.xidian.edu.cn}.}
% \thanks{Dimitri Lefebvre is with GREAH Laboratory, Normandy University, 75 rue Bellot, Le Havre 76600, France, E-mail:
%         {\tt\small dimitri.lefebvre@univ-lehavre.fr}.}
% \thanks{Carla Seatzu is with DIEE, University of Cagliari, Cagliari 09124, Italy, E-mail:
%         {\tt\small carla.seatzu@unica.it}.}
% \thanks{Zhiwu Li is with the Institute of Systems Engineering, Macau University of Science and Technology, Macau, E-mail:
% {\tt\small zwli@must.edu.mo}.}
% \thanks{Alessandro Giua is with DIEE, University of Cagliari, Cagliari 09124, Italy, E-mail:
%         {\tt\small giua@unica.it}.}%
% }

\maketitle

\begin{abstract}
%Many industrial control systems can be described as discrete event systems whose state space is a discrete set, where the transitions from one state to another are caused by event occurrences.
In this paper, we consider partially observable timed automata endowed with a single clock.
A time interval is associated with each transition specifying at which clock values it may occur. In addition, a resetting condition associated to a transition specifies how the clock value is updated upon its occurrence.
%We assume that the (logical and timed) structure of a partially observable timed automaton is known.
%The behaviour of a timed automaton can be described via its timed runs. A timed observation is present in a succession of pairs of an observable symbol and the time instant at which the event occurs and a time instant at which the observation ends.
%{\red A (hybrid/timed) state of a timed automaton consists in a discrete state given by the logical structure and a clock value specified by the timed structure.
This work deals with the estimation of the current state given a timed observation, i.e., a succession of pairs of an observable event and the time instant at which the event has occurred.
%We study the evolution of a timed automaton which consists in the interleaving of time-driven and event-driven steps.
%A finite state automaton, called a zone automaton, is associated with the timed automaton to provide a purely discrete event description of its behaviour.
%{\red The state space of a zone automaton is a set of extended states, which are pairs of a discrete state and a time interval.
%Based on constructing a zone automaton that provides a purely discrete event description of the behaviour of an associated timed automaton.
The problem of state reachability in the timed automaton is reduced to the reachability analysis of the associated zone automaton, which provides a purely discrete event description of the behaviour of the timed automaton.
An algorithm is formulated to provide an offline approach for state estimation of a timed automaton based on the assumption that the clock is reset upon the occurrence of each observable transition.
\end{abstract}


\section{Introduction}
\label{sec:introduction}
%\IEEEPARstart{T}{he} goal of this paper is that of designing an observer for estimating the state of a discrete event system modelled by a finite state automaton, when the occurrence of its events is constrained by a timed structure.

In the area of \emph{discrete event systems} (DES), a large body of literature considers logical DES, where time is abstracted and only the order of occurrences of the events is taken into account. %Different models can be used to describe these systems, including, \emph{finite state automata} and \emph{Petri nets}.
The problem of state estimation has received a lot of attention considering different observation structures %, e.g., assuming that only a subset of the event occurrences or, possibly, also a function of the state, can be measured. Contributions in state estimation have been proposed
in different formalisms, in particular automata \cite{hadjicostis2020estimation} and Petri nets \cite{Cabasino2010}.
%automata \cite{Ramadg1986,OzverenWillsky1990,Kumar1993,hadjicostis2020estimation} and Petri nets \cite{GiuaSeatzu2002,giua2007marking,Jiroveanu2008,RuHadjicostis2009,Cabasino2010,tong2015equivalence}.
%Most of these works are based on the notion of \emph{observer} \cite{CL}, namely a new DES that that contains all the information to reconstruct the set of current states consistent with an observation.
The notion of \emph{observer} \cite{CL} contains all the information to reconstruct the set of current states consistent with an observation and allows to move most of the burdensome parts of the computations offline and can be used for online state estimation.
In addition, an observer is fundamental to address problems of feedback control \cite{yin2015uniform}, diagnosis and diagnosability \cite{sampath1996failure}, \cite{sampath1995diagnosability}, in addition to characterizing a large set of dynamical properties, such as detectability  \cite{shu2007detectability}, opacity  \cite{lin2011opacity} and resilience to cyber-attack \cite{carvalho2018detection}.

%From the aspect of state estimation of timed DES,
%\cite{Miao2017} investigates the mechanism of communication delays in observations and illustrates all the possible observed trajectories using an augmented automaton.
Significant contributions have also been provided in the framework of timed DES.
As an example, in \cite{Laiaiwen2019}, observers are designed for a particular class of weighted automata, where the time information of transition firability is given to weights.
%In \cite{Zhangkuize2021}, observers are proposed to verify current-state opacity of real-time automata.
%However, the observer may not be unique (in general) due to various selection of the set of events in the observer.
References \cite{Lijun2022}, \cite{gao2020etfa}, \cite{gao2023ifac}, and \cite{lefebvre2023automatica} address the problem of state estimation of a class of timed automata under a rather restrictive scenario where the endowed single clock is reset to zero after each event occurrence.
The state estimation of timed DES has also been considered in \cite{Basile2017} and \cite{He2019}, but no general approach concerning the construction of an observer for a general class of timed DES exists.

%the mechanism of communication delays in observations is investigated, and an algorithm to obtain state estimation under partial observation is proposed via constructing an augmented automaton that illustrates all the possible observed trajectories \cite{Miao2017}.
%Observers are designed for a particular class of weighted automata, presented by Max-plus automata, which are strongly related to timed automata if a timed interpretation, i.e., the minimal time required by the process to fire the transitions, is given to weights \cite{Laiaiwen2019}.
%State-based opacity of real-time automata is investigated and an observer (resp., a reverse observer) is proposed to concatenate current state estimates along timed output sequences generated by the real-time automata and to verify current-state opacity (resp., to verify initial-state opacity) \cite{Zhangkuize2021}.
%However, the observer may not be unique (in general) due to various selection of the set of events in the observer.
%A class of timed and labeled finite automata called constant-time labeled automata is considered in \cite{Lijun2022}, where the time information provided by time stamps can be used to refine the outcome of state estimation in various situations.
%In \cite{Lijun2022}, events occur according to a single clock at constant times under certain time semantics.
%In \cite{gao2020etfa} and \cite{gao2023ifac}, a class of timed automata characterized by a single clock is considered.
%The state estimation is addressed under a rather restrictive scenario that the clock is reset to zero after each event occurrence.

Another active area of research is that of \emph{hybrid systems} (HS) \cite{alur1995algorithmic,henzinger2000theory},
%\cite{alur1995algorithmic},\cite{henzinger2000theory},
characterized by the interplay between discrete event and time-driven dynamics.
These systems can be modeled by \emph{hybrid automata}, and in particular by \emph{timed automata} \cite{c1} concerning time elapsing as time-driven dynamics.
%a very general formalism for which, many properties of interest are undecidable.
%This motivates several authors working in this domain to explore several subclasses of models with restricted continuous dynamics, such as the class of \emph{timed automata} proposed by Alur and Dill \cite{c1}, for which some fundamental properties, such as \emph{reachability}, are decidable.
%While some other problems are known to be undecidable for timed automata \cite{Alur2004,Finke2006,Bouyer2009}, tractable subclasses with a restricted number of clocks \cite{Hung1996,Alur1999,Laroussinie2004,Dima2001} have been identified.
State estimation of timed automata has also been studied.
In \cite{c1}, the reachability of locations can be analysed by searching the finite quotient of a timed automaton with respect to the region equivalence defined over the set of all clock interpretations.
However, the reachability is analyzed regardless of any observations.
Tripakis proposes an online diagnoser in \cite{Tripakis2002} that keeps track of all the possible discrete states.
Thereby, state estimation problem is theoretically solvable.
In \cite{Bouyer2018} and \cite{Bouyer2021}, timed markings are used for representing the closure under silent transitions.
%, which can be precomputed, such that the state estimation approaches are more efficient.
Observe that in \cite{Tripakis2002,Bouyer2018,Bouyer2021}, only online estimators have been proposed, and no general approach concerning the construction of an observer is known.
%Hybrid automata are a generalization of finite state automata, and thus a timed DES can also be seen as a restricted type of HS.
In particular the model we consider in this paper, that we call \emph{timed finite automaton} can be either seen as a finite state automaton endowed with a clock or as a timed automaton \cite{c1} whose edges are labeled.

%studies the problem of fault diagnosis and proposes a method to check diagnosability in \cite{Tripakis2002}, where an online diagnoser can be implemented utilizing an online state estimator that keeps track of all the possible discrete states and the associated clock constraints after each event is measured after a delay. Thereby, state estimation problem is theoretically solvable.
%State estimation of timed automata with a single clock is explored in \cite{Bouyer2018} using timed markings that represent the closure under silent transitions. It is shown that such closure can be performed as a series of subtractions between an associated interval and a regular union of intervals that can be precomputed, and the precomputation leads the proposed approach to be more efficient compared with the online approach in \cite{Tripakis2002}.
%Based on \cite{Bouyer2018}, the work \cite{Bouyer2021} provides some insights to state estimation of timed automata with multiple clocks.
%We mention that while the state estimation of labeled timed automata has been studied in the HS literature \cite{Tripakis2002,Bouyer2018,Bouyer2021}, only online estimators have been proposed, and no general approach concerning the construction of an observer is known.

A related problem in a decentralized setting is studied in \cite{giua2017decentralized}, where the asynchronous polling of distributed sub-systems is called \emph{synchronization}.
Concerning timed automata, references \cite{DSouza2000}, \cite{Komenda2010}, and \cite{Lin2019} discuss how to define concurrent composition for timed automata, which can be used to construct complex models.

%In \cite{DSouza2000}, product interval automata are introduced as synchronous products of interval automata, where each interval automaton operates with a single clock that is reset after each transition.
%The synchronous composition of interval automata is investigated in \cite{Komenda2010} using techniques based on semirings of interval and tensor linear algebra.
%In \cite{Lin2019}, the semantic interpretation of the model of finite interval automata is formalized in terms of finite tick automata.
%These works may inspire extension to timed automata with multiple clocks using synchronization.

This paper considers a partially observable timed finite automaton (TFA) endowed with a single clock.
The logical structure specifies the sequences of events that the TFA can generate and the observations they produce.
The timed structure specifies the set of clock values that allow an event to occur and how the clock is reset upon the event occurrence.
A timed state of a TFA consists of a discrete state associated with the logical structure and the clock value.
Our objective is that of estimating the current discrete state of the automaton as a function of the current observation and of the current time.
The notion of $T$-reachability is proposed taking into account which discrete states can be reached with an evolution that produces a given observation and has a duration equal to $T$.
We prove that the problem of $T$-reachability for a TFA is reduced to the reachability analysis of the associated \emph{zone automaton}.
We assume that the firing of observable transitions resets the clock to a set that is a finite union of zones, and present a formal approach that allows one to construct offline an observer, i.e., a finite structure that describes the state estimation for all possible evolutions.
During the online phase to estimate the current discrete state, one just has to determine which is the state of the observer reached by the current observation and check to which interval (among a finite number of time intervals) the time elapsed between the last observed event occurrence belongs.
We believe that the proposed approach has a major advantage over existing online approaches for state estimation: it paves the way to address a vast range of fundamental properties (detectability, opacity, etc.) that have so far mostly been studied in the context of logical DES.

%The solution proposed in this paper is based on the notion of $T$-reachability, which takes into account which discrete states can be reached with an evolution that produces a given observation and has a duration equal to $T$.
%The problem of $T$-reachability for a TFA is reduced to the reachability analysis of the associated \emph{zone automaton}, which provides a purely discrete event description of the behaviour of the TFA by extending the work in \cite{gao2023ifac}.
%We assume that the firing of observable transitions nondeterministically resets the clock to a set that is a finite union of zones.
%Under this assumption --- that we call ``reinitialized observations'' --- we present a formal approach that allows one to construct offline an observer, i.e., a finite structure that describes the state estimation for all possible evolutions.
%During the online phase to estimate the current discrete state, one just has to determine which is the state of the observer reached by the current observation and check to which interval (among a finite number of time intervals) the time elapsed between the last observed event occurrence belongs.
%We believe that the approach we present in this work has a major advantage over existing online approaches for state estimation: it paves the way to address a vast range of fundamental properties (detectability, opacity, etc.) that have so far mostly been studied in the context of logical DES.

The rest of the paper is organized as follows.
Section \uppercase\expandafter{\romannumeral2} introduces the background of DES, timed automata and time semantics.
Section \uppercase\expandafter{\romannumeral3} formally state the problem of state estimation.
Section \uppercase\expandafter{\romannumeral4} illustrates the approaches to compute zones and to construct the
zone automaton.
Section 
\uppercase\expandafter{\romannumeral5} provides necessary and sufficient conditions for the reachability of a state of the timed automaton in terms of reachability analysis in the zone automaton.
Section \uppercase\expandafter{\romannumeral6} solves the problem of updating the state estimation according to the latest received {\em timed observation}.
Finally, Section \uppercase\expandafter{\romannumeral7} concludes this paper.

\section{Background}
%A \textit{deterministic finite automaton (DFA)} is a four-tuple $G = (X, E, \delta, x_0)$, where $X$ is a finite set of states, $E$ is a finite set of events, $\Delta: X\times E \times X$ is a transition relation and $x_0\in X$ is the initial state.
A nondeterministic finite automaton (NFA) is a four-tuple $G = (X, E, \Delta, X_0)$, where $X$ is a finite set of states, $E$ is a finite set of events, $\Delta\subseteq X\times
E \times X$ is the transition relation and $X_0 \subseteq X$ is the set of initial states.
The set of events $E$ can be partitioned as $E = E_o \dot\cup E_{uo}$, where $E_o$ is the set of observable events, and $E_{uo}$ is the set of unobservable events.
We denote by $E^\ast$ the set of all finite strings on $E$, including the empty word $\varepsilon$.
%Given a state $x\in X$, the set of \emph{input and output transitions} of $x$ is $\Delta_{in}(x)=\{(x^\prime,e,x)\in \Delta| e\in E, x^\prime\in X\}$ and $\Delta_{out}(x)=\{(x,e,x^\prime)\in \Delta| e\in E, x^\prime\in X\}$.
The \emph{concatenation} $s_1\cdot s_2$ of two strings $s_1\in E^\ast$ and $s_2\in E^\ast$ is a string consisting of $s_1$ immediately followed by $s_2$. The empty string
$\varepsilon$ is an identity element of concatenation, i.e., for any string $s\in E^\ast$, it holds that $\varepsilon\cdot s= s = s\cdot \varepsilon$.
Given a string in $E^\ast$, the observation is defined via the observation projection $P_l: E^\ast \longrightarrow E_o^\ast$ defined as:
$P_l(\varepsilon)=\varepsilon $, and for all $s\in E^\ast$ and $e \in E$, it is $P_l(s\cdot e) = P_l(s)\cdot e$ if $e\in E_o$, or $P_l(s\cdot e)=P_l(s)$ if $e\in E_{uo}$.

Let $\mathbb{R}_{\geq 0}$ be the set of non-negative real numbers and $\mathbb{N}$ be the set of natural numbers.
% {\blue
% A \emph{time interval} $I$ is a set of real numbers lying between a lower bound $l(I)=m\in \mathbb{N}$ and an upper bound $u(I)=n\in \mathbb{N}\cup \{+\infty\}$ ($m\leq n$).
% }
%The set of real numbers lying between a lower bound $m\in \mathbb{N}$ and an upper bound $n\in \mathbb{N}\cup \{+\infty\}$ ($m\leq n$) is said to be a \emph{time interval}.
%Such a time interval $I$ is left-closed (resp., right-closed) if $m \in I$  (resp., $n\in I$), else it is left-open (resp., right-open).
A closed interval, i.e., closed on both sides, is denoted as $[m,n]$, while open or semi-open intervals are denoted as $(m,n)$,  $[m,n)$ or $(m,n]$.
%The lower bound $m$ (resp., upper bound $n$) is denoted as $I_{lower}=m$ (resp., $I_{upper}=n$).
We denote the set of all time intervals and the set of all closed time intervals as $\mathbb{I}$ and $\mathbb{I}_c$, respectively, where $\mathbb{I}_c \subseteq \mathbb{I}$.
%{\blue
%The \emph{set of regions} of $[h,k]\in \mathbb{I}_c$ is defined as $Regions([h,k])=\{[h,h],(h,h+1),[h+1,h+1], \cdots, [k,k]\}$.
%Given two time intervals $I_1,I_2\in \mathbb{I}$, $I_1\prec I_2$ if $t_1 < t_2$ holds for each $t_1\in I_1$ and $t_2\in I_2$.
%}
The \emph{addition}\footnote{The addition operation is associative and commutative and can be extended to $n >2$ time intervals $\bigoplus\limits_{i=1}^n I_i = I_1 \bigoplus \cdots \bigoplus I_n$.} of $I_1$ and $I_2$ is defined as $I_1 \bigoplus I_2 =\{t_1+t_2\in \mathbb{R}_{\geq 0} \mid t_1\in I_1, t_2\in I_2\}$ and the \emph{distance range} between them as $D(I_1,I_2)=\{ |t_1 - t_2 | \ \mid t_1\in I_1, t_2\in I_2 \}$.
\begin{definition}
A \emph{timed finite automaton} (TFA) is a six-tuple $G=(X$, $E$, $\Delta$, $\Gamma$, $Reset$, $X_0)$ that operates under a single clock, where $X$ is a finite set of discrete states, $E$ is an alphabet, $\Delta\subseteq X\times E \times X$ is a transition relation, $\Gamma: \Delta \rightarrow
\mathbb{I}_c$ is a timing function, $Reset: \Delta \rightarrow \mathbb{I}_c\cup \{id\}$ is a clock resetting function such that for $\delta\in\Delta$, the clock is reset to be an integer value in a time interval $I\in \mathbb{I}_c$ ($Reset(\delta)= I$), or the clock is not reset ($Reset(\delta)= id$), and $X_0 \subseteq X$ is the set of initial discrete states.\hfill$\square$
\end{definition}

%In other words, a TFA $G=(X, E, \Delta, \Gamma, Reset, X_0)$ is a NFA $G=(X, E, \Delta, X_0)$ endowed with a time structure represented by the timing function $\Gamma$, and the clock resetting function $Reset$.
For the sake of simplicity, we assume that the clock is set to be 0 initially.
A transition $(x,e,x^\prime)\in \Delta$ denotes that the occurrence of event $e\in E$ leads to $x^\prime\in X$ when the TFA is at $x\in X$.
The time interval $\Gamma((x,e,x^\prime))$ specifies a range of clock values at which the event $e$ may occur, while $Reset((x,e,x^\prime))\in \mathbb{I}_c$ denotes the range of the clock values that are reset to be and $Reset((x,e,x^\prime))=id$ implies that the clock is not reset.
The \emph{set of output transitions at $x$} is defined as 
$O(x)=\{(x,e,x^\prime)\in \Delta \mid  e\in E, x^\prime \in X \}$, and \emph{the set of input transitions at $x$} is defined as $I(x)=\{(x^\prime,e,x)\in \Delta \mid  e\in E,  x^\prime \in X\}$.

A \emph{timed state} is defined as a pair $(x,\theta)\in X\times \mathbb{R}_{\geq 0}$, where $\theta$ is the current value of the clock.
In other words, a timed state $(x,\theta)$ keeps track of the current clock assignment $\theta$ while $G$ stays at state $x$.
%\footnote{ According to the usual terminology in hybrid systems community, the extended state $(x, \theta)$ is the hybrid state of the timed automaton, while $x$ and $\theta$ are the discrete and the continuous states of the timed automaton, respectively.}
The behaviour of a TFA is described via its timed runs.
A \emph{timed run} $\rho$ of length $k\geq 0$ from $t_0\in \mathbb{R}_{\geq 0}$ to $t_k\in \mathbb{R}_{\geq 0}$ is a sequence of $k+1$ timed states $(x_{(i)},\theta_{(i)})\in X\times \mathbb{R}_{\geq 0}\ (i =0,\cdots,k)$, and $k$ pairs $(e_i,t_i)\in E \times \mathbb{R}_{\geq 0}\ (i =1,\cdots,k)$, represented as
$\rho: (x_{(0)},\theta_{(0)}){\xrightarrow{(e_1,t_1)}}\cdots(x_{(k-1)},\theta_{(k-1)}){\xrightarrow{(e_k,t_k)}}(x_{(k)},\theta_{(k)})$
such that $(x_{(i-1)},e_i,x_{(i)})\in \Delta$, 
and the following conditions hold for all $i =1,\cdots,k$:
\begin{itemize}
    \item $\theta_{(i)}\in Reset((x_{(i-1)},e_i,x_{(i)}))$ and $\theta_{(i-1)}+ t_i-t_{i-1}\in \Gamma((x_{(i-1)},e_i,x_{(i)}))$, if $Reset((x_{(i-1)},e_i,x_{(i)}))\neq id$;
    \item $\theta_{(i)} = \theta_{(i-1)}+ t_i-t_{i-1}\in \Gamma((x_{(i-1)},e_i,x_{(i)}))$, if $Reset((x_{(i-1)},e_i,x_{(i)}))= id$.
\end{itemize}

%the following conditions are satisfied for all $i =1,\cdots,k$:
%\begin{equation}(x_{(i-1)},e_i,x_{(i)})\in \Delta,\end{equation}
%\begin{equation}\theta_{(i)}\in Reset((x_{(i-1)},e_i,x_{(i)})),\end{equation}
%\begin{equation}\label{Gammaconstraint} \theta_{(i-1)}+ t_i-t_{i-1}\in \Gamma((x_{(i-1)},e_i,x_{(i)})).\end{equation}

We define the \emph{timed word generated by $\rho$} as $\sigma(\rho)=(e_1, t_1)(e_2, t_2)\cdots(e_k$, $t_k) \in (E\times \mathbb{R}_{\geq 0})^\ast$. We also define the \emph{logical word generated by
$\rho$} as $S(\sigma(\rho))= e_1 e_2\cdots e_k$ via a function defined as $S:(E \times \mathbb{R}_{\geq 0})^\ast \rightarrow E^\ast$.
For the timed run of length 0 as $\rho:(x_{(0)}, \theta_{(0)})$, we have $S(\sigma(\rho))= \varepsilon$ and
$\sigma(\rho)=\lambda$, where $\lambda$ denotes the \emph{empty timed word} in $E \times \mathbb{R}_{\geq 0}$.
For the timed word $\sigma(\rho)$ generated from an arbitrary timed run $\rho$, it is $\lambda\cdot\sigma(\rho)= \sigma(\rho)= \sigma(\rho)\cdot\lambda$.
The \emph{starting discrete state} and the \emph{ending discrete state} of a timed run $\rho$ are denoted by $x_{st}(\rho)= x_{(0)}$ and $x_{en}(\rho)= x_{(k)}$, respectively.
%{\red
%The \emph{starting state} and the \emph{ending state} of a timed run $\rho$ are denoted by $s_{st}(\rho)= (x_{(0)},\theta_{(0)})$ and $s_{en}(\rho)= (x_{(k)},\theta_{(k)})$, respectively.
%}
The \emph{starting time} and the \emph{ending time} of $\rho$ are denoted by $t_{st}(\rho)=t_0$ and $t_{en}(\rho)=t_k$, respectively.
In addition, the {\em duration of $\rho$} is denoted as $T(\rho)=t_k-t_0$.
%and the \emph{ending clock value} is denoted as $\theta_{en}(\rho)=\theta_{(k)}$.
%{\red Note that Eq.~\eqref{Gammaconstraint} clearly implies $T(\rho) \in \bigoplus\limits_{i=0}^{k-1}\Gamma(x_{(i)},e_{i+1},x_{(i+1)}).$}
The set of timed runs generated by $G$ is denoted as $\mathcal{R}(G)$.

%Ö»¿¼ÂÇÒ»ÖÖsemantics

A {\em timed evolution} of $G$ from $t_0\in\mathbb{R}_{\geq 0}$ to $t\in\mathbb{R}_{\geq 0}$ is defined by a pair $(\sigma(\rho),t) \in (E\times \mathbb{R}_{\geq 0})^\ast \times \mathbb{R}_{\geq 0}$, where $t_{st}(\rho)=t_0$ and $t_{en}(\rho)\leq t$.
%In other words, a timed evolution of $G$ from $t_0$ to $t$ is defined as a pair whose first entry is a timed word $\sigma(\rho)$, where $\rho$ is a timed run that starts at $t_0$ from an initial state $X_0$ and ends at $t_{en}(\rho)\ (t_{en}(\rho)\leq t)$, and whose second entry is the time instant $t$.
%$$
%\begin{array}{rcl}
%\mathcal{E}(G,t)=\{(\sigma(\rho), t) \mid &(\exists \rho\in \mathcal{R}(G))\ x_{st}(\rho)\in X_0, \\
%                                          & t_{st}(\rho)=0,t_{en}(\rho)\leq t\}
%\end{array}
%$$
Note that $t-t_{en}(\rho)$ is the time that the system stays at the ending discrete state $x_{en}(\rho)$.
Furthermore, we denote as $\mathcal{E}(G,t)=\{(\sigma(\rho), t) \mid (\exists \rho\in \mathcal{R}(G))\ x_{st}(\rho)\in X_0, t_{st}(\rho)=0, t_{en}(\rho)\leq t\}$
the \emph{timed language of $G$ from $0$ to $t\in \mathbb{R}_{\geq 0}$}, that contains all possible timed evolutions of $G$ from $0$ to $t$.

%and the current clock value associated with $(\sigma(\rho),t)$ is $\theta_{en}(\rho)+(t-t_{en}(\rho))$.
%In addition, the timed language of $G$ from $0$ to $t$ contains all possible timed evolutions of $G$ from $0$ to $t$.

\begin{definition}
Given a TFA $G=(X, E, \Delta, \Gamma, Reset, X_0)$ and a time instant $T\in\mathbb{R}_{\geq 0}$, a discrete state $x^\prime\in X$ is said to be \emph{$T$-reachable} from $x\in X$ if there exists a timed evolution $(\sigma(\rho),t)\in (E\times \mathbb{R}_{\geq 0})^\ast \times \mathbb{R}_{\geq 0}$ of $G$ such that $t-t_{st}(\rho)=T$, $x_{st}(\rho)=x$, and $x_{en}(\rho)= x^\prime$.
In addition, $x^\prime$ is said to be \emph{unobservably $T$-reachable} from $x$ if $x^\prime$ is $T$-reachable from $x$ with a timed evolution $(\sigma(\rho),t)$ such that $S(\sigma(\rho))\in E_{uo}^\ast$.\hfill$\square$
\end{definition}

In simple words, $x^\prime$ is $T$-reachable from $x$ if a timed evolution leads the system from $x$ to $x^\prime$ with an elapsed time $T$.
If there exists such a timed evolution that produces no observation, $x^\prime$ is unobservably $T$-reachable from $x$.

\begin{example}
Given a TFA $G=(X$, $E$, $\Delta$, $\Gamma$, $Reset$, $X_0)$ in Fig.~\ref{tfaG}(a) with $X = \{x_0, x_1, x_2, x_3, x_4\}$, $E = \{ a, b, c\}$, $\Delta=\{(x_0, c, x_1)$, $(x_0, b, x_2)$, $(x_1, a, x_4)$, $(x_2, c, x_3)$, $(x_3, a, x_2)$, $(x_4, b, x_3)\}$, the initial state in $X_0=\{x_0\}$ is marked by an input arrow.
%The graphical representation of $G$ is visualized in Fig.~\label{tfaG}(b), where a discrete state $x\in X$ corresponds to a node, and each initial discrete state in $X_0$ is marked by an input arrow.
%For any transition $(x,e,x^\prime)\in \Delta$, there exists a directed edge from $x$ to $x^\prime$ labeled with the symbol $e$.
%The information given by the timing function $\Gamma$ and the clock resetting function $Reset$ can be presented graphically.
The information given by the timing function $\Gamma$ and the clock resetting function $Reset$ defined in Fig.~\ref{tfaG}(b) is presented on the edges.
Given an edge denoting a transition $\delta\in \Delta$, the label $\theta\in\Gamma(\delta)?$ on the edge specifies if $\delta$ is enabled with respect to $\theta$; the label $\theta:\in Reset(\delta)$ (resp., $\theta:= id$) on the edge specifies to which range $\theta$ belongs (resp., specifies that the clock is not reset) after the transition is fired.
%Consider a timed run of length 3 from time $0$ to $3$ as $\rho: (x_{0},0)\stackrel{(c,1.5)}{\longrightarrow}(x_{1},1)\stackrel{(a,3)}{\longrightarrow}(x_{4},0.6)$.
%The timed word $\sigma(\rho)=(c,1.5)(a,3)$ corresponds to events $c$ and $a$ occurring at time instants $t_1=1.5$ and $t_2=3$, respectively.
%It starts from $x_{st}(\rho) =x_0$ at the starting time $0$ and terminates in $x_{en}(\rho) =x_4$ at the ending time $3$.
%The logical word generated by $\rho$ is $S(\sigma(\rho))=ca$.
%It involves two transitions, namely $(x_0,c,x_1)$ and $(x_1,a,x_4)$, leading the clock reset to be $1\in Reset((x_0,c,x_1))$ at $x_1$ and to be $0.6\in Reset ((x_1,a,x_4))$ at $x_4$, respectively.
%In addition, we have $0+ t_1\in \Gamma((x_0,c,x_1))$ and $1+ t_2- t_1\in \Gamma((x_1,a,x_4))$.
Consider a timed run $\rho: (x_{0},0)\stackrel{(b,0.5)}{\longrightarrow}(x_{2},0.5)\stackrel{(c,2)}{\longrightarrow}(x_{3},2)\stackrel{(a,2)}{\longrightarrow}(x_{2},0)$ that starts from $x_{st}(\rho) =x_0$ at $t_{st}(\rho)=0$ and terminates in $x_{en}(\rho) =x_2$ at $t_{en}(\rho)=2$.
Three transitions $(x_0,b,x_2)$, $(x_2,c,x_3)$, and $(x_3,a,x_2)$ occur at time instants $t_1=0.5$, $t_2=2$ and $t_3=2$, respectively.
The transitions $(x_0,b,x_2)$ and $(x_2,c,x_3)$ do not lead the clock to be reset, while $(x_3,a,x_2)$ resets the clock to $0$.
Given a timed evolution $(\sigma(\rho),2)$, it follows that $x_2$ and $x_3$ are $2$-reachable from $x_0$.\hfill$\square$
\end{example}

\begin{figure}[h]
\subfigure[A TFA $G$.]{
\begin{minipage}{0.5\linewidth}
\includegraphics[width=5.5cm]{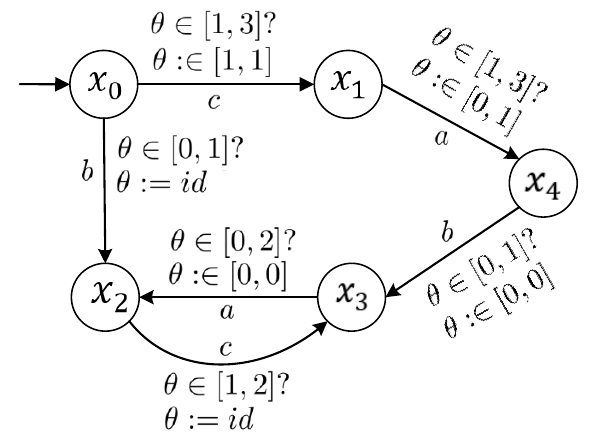}\\
%\caption{A TFA $G$.}\label{tfaG}
\end{minipage}
}
\hfill
\subfigure[Timing function and clock resetting function.]{
\begin{minipage}{0.45\linewidth}
\centering
\scalebox{0.8}{
\begin{tabular}{c|c|c}
\hline
$\delta\in\Delta$ & $\Gamma(\delta)$ & $Reset(\delta)$\\
\hline
$(x_0, c, x_1)$ & $[1,3]$ & $[1,1]$\\
$(x_0, b, x_2)$ & $[0,1]$ & $id$\\
$(x_1, a, x_4)$ & $[1,3]$ & $[0,1]$\\
$(x_2, c, x_3)$ & $[1,2]$ & $id$\\
$(x_3, a, x_2)$ & $[0,2]$ & $[0,0]$\\
$(x_4, b, x_3)$ & $[0,1]$ & $[0,0]$\\
\hline
\end{tabular}
}
%\caption{The timing function and the clock resetting function of the TFA $G$ in Fig.~\ref{tfaG}.}
\end{minipage}
}
\caption{A TFA $G$ w.r.t. the given timing function and clock resetting function.}\label{tfaG}
\end{figure}

\section{Problem Statement}
In this work we model a partially observed timed plant as a TFA $G=(X, E, \Delta, \Gamma, Reset, X_0)$ with a partition of the alphabet into observable and unobservable events: $E=E_o\dot\cup E_{uo}$.
%{\red
%For the sake of simplicity, we make the following assumption:
%\begin{enumerate}
%  \item [(A1)] the TFA is nonZeno, i.e., the system cannot enforce infinite transitions in a finite and bounded length of time.
%\end{enumerate}
%}
Next we preliminarily define a \emph{projection function} on timed words.

\begin{definition}
Given a TFA $G$ with $E=E_o\dot\cup E_{uo}$, a \emph{projection function} $P : (E\times \mathbb{R}_{\geq 0})^\ast \longrightarrow (E_o\times
\mathbb{R}_{\geq 0})^\ast$ is defined as
$P(\lambda)=\lambda$, and
$P(\sigma(\rho)\cdot(e,t)) = P(\sigma(\rho))$ if $e\in E_{uo}$, or $P(\sigma(\rho)\cdot(e,t)) = P(\sigma(\rho))\cdot(e,t)$ if $e\in E_{o}$,
%$$
%P(\sigma(\rho)\cdot(e,t))=\left\{
%\begin{array}{rcl}
%P(\sigma(\rho))   && {\mbox{if} \  e\in E_{uo}}\\
%P(\sigma(\rho))\cdot(e,t) && {\mbox{if} \  e\in E_{o}}
%\end{array}\right.
%$$
for the timed word $\sigma(\rho)\in (E\times \mathbb{R}_{\geq 0})^\ast$ generated from any timed run $\rho\in\mathcal{R}(G)$ and for all $(e, t)\in E\times \mathbb{R}_{\geq 0}$.
Given a timed evolution $(\sigma(\rho),t)$, the pair $(\sigma_o, t)= (P(\sigma(\rho)),t)$ is said to be the \emph{timed observation}.
\hfill$\square$
\end{definition}

\begin{definition}
Given a TFA $G$ with $E=E_o\dot\cup E_{uo}$, and a \emph{timed observation} $(\sigma_o, t)$,
$\mathcal{S}(\sigma_o, t)=\{(\sigma(\rho), t)\in \mathcal{E}(G,t)| P(\sigma(\rho))= \sigma_o \}$
is said to be the\emph{ set of timed evolutions consistent with $(\sigma_o, t)$}, i.e., the set of timed evolutions that can be generated by $G$ from $0$ to $t$ producing the timed observation $(\sigma_o, t)$; meanwhile
$\mathcal{X}(\sigma_o, t)=\{x_{en}(\rho) \in X |(\sigma(\rho),t)\in \mathcal{S}(\sigma_o, t) \}$
is said to be \emph{the set of discrete states consistent with $(\sigma_o, t)$}, i.e., the set of discrete states in which $G$ may be, after $(\sigma_o, t)$ is observed.\hfill$\square$
\end{definition}

%As shown in Fig.~\ref{statement},
This work aims at calculating the set $\mathcal{X}(\sigma_o, t)$, which includes the discrete states reached by each timed evolution $(\sigma(\rho),t)$ consistent with $(\sigma_o, t)$. In addition, this work also provides a range of the possible clock values associated with each timed evolution $(\sigma(\rho),t)$.

%In this work, given a timed observation $(\sigma_o, t)$ generated by a partially observed TFA, our goal consists in determining the discrete states in which the TFA can be and the range of the possible clock values associated with each estimated discrete state.

%Looking again at Fig.~\ref{statement}, the set $\mathcal{X}(\sigma_o, t)$ is indeed the output of the state estimation process.

\section{Zone automaton}

%The work \cite{gao2023ifac} explores timed automata with a single clock that is reset at each event occurrence and introduces the notion of \emph{zone automaton}, which is a finite state automaton providing a purely discrete event description of the behaviour of such a timed automaton.

The notion of \emph{zone automaton} has been introduced in \cite{gao2023ifac} to provide a purely discrete event description of the behaviour of a given timed automaton endowed with a single clock reset at each event occurrence and bounded by the maximal dwell time at each discrete state.
The zones associated with a discrete state $x$ partition the clock values at $x$.
The timed automaton in \cite{gao2023ifac} can be seen as a particular case of the TFA in this paper.
In this section, a new algorithm is provided to compute the zones associated with a given discrete state.
After that, we illustrate how to construct the zone automaton based on the zone automaton in \cite{gao2023ifac}.
We first introduce the following definitions.

\begin{definition}
\label{def:regions}
The \emph{set of regions} of a discrete state $x\in X$ is defined as $R(x)=\{[m_x,m_x], (m_x,m_x+1), [m_x+1,m_x+1], (m_x+1,m_x+2), \cdots$, $[M_x,M_x] \}$,
where $m_x$ (resp., $M_x$) is the minimal (resp., maximal) integer in 
$\{\Gamma(\delta)|\delta\in O(x)\vee (\delta\in I(x), Reset(\delta)=id)\}\cup \{Reset(\delta)\neq id|\delta\in I(x)\}. \hfill\square$
%that appears in the time intervals mapping by the timing function and clock resetting function from $x$.  $\hfill\square$
% $$
% \begin{array}{cc}
% m=\min\{ &\{l(\Gamma (\delta))\mid  \delta \in \mathcal{O}(x) \},\\
% &\{l(Reset(\delta))\mid \delta\in \mathcal{I}(x), Reset(\delta)\neq id\}, \\
% &\{ l(\Gamma(\delta))\mid \delta \in \mathcal{I}(x), Reset(\delta)=id\}\}, %\textit{and}
% \end{array}
% $$
% $$
% \begin{array}{cc}
% M=\max\{  
% &\{u(\Gamma (\delta))\mid \delta \in \mathcal{O}(x)\},\\
% &\{u(Reset(\delta))\mid \delta\in \mathcal{I}(x), Reset(\delta)\neq id \},\\
% &
% \{u(\Gamma(\delta))\mid Reset(\delta)=id\}
% \}.\hfill\square

% \end{array}
% $$
\end{definition}

The integer $m_x$ (resp., $M_x$) represents the minimal (resp., maximal) clock value that can enable an output transition at $x$ and that can reach $x$ by an input transition.
Note that for a transition $\delta$ inputting state $x$, we search minimal (resp., maximal) value in $\Gamma(\delta)$ if $\delta$ does not reset the clock or $Reset(\delta)$ if $\delta$ resets the clock.   
The regions of $x$ include the integers from $m_x$ to $M_x$ and the open segments between the integers.
Given $r=[k,k]\in R(x)$ (resp., $r=(k,k+1)\in R(x)$), where $k=m_x,\cdots, M_x-1$, its successive region is denoted as $succ(r)=(k,k+1)$ (resp., $succ(r)=[k+1,k+1]$).
For instance, given the discrete state $x_0$ of the TFA $G$ in Fig.~\ref{tfaG}, we have $m_{x_0}=0$ , $M_{x_0}=3$, and $R(x_0)=\{[0,0], (0,1), [1,1], \cdots, [3,3]\}$.

\begin{definition}
Given a TFA $G=(X$, $E$, $\Delta$, $\Gamma$, $Reset$, $X_0)$, \emph{the set of output transitions at $(x,r)\in X\times \bigcup\limits_{x\in X} R(x)$} is defined as 
$O(x,r)=\{(x,e,x^\prime)\in \Delta |  e\in E, x^\prime \in X, r\subseteq \Gamma((x,e,x^\prime))\}$,
and \emph{the set of input transitions at $(x,r)$} is defined as 
$I(x,r)= 
\{(x^\prime,e,x)\in \Delta \mid e\in E, (r\subseteq  Reset((x^\prime,e,x))\neq id)\vee (Reset((x^{\prime},e,x))=id, r\subseteq \Gamma ((x^{\prime},e,x)))\} $.
Obviously, $O(x)=\bigcup\limits_{r\in R(x)}O(x,r)$ and $I(x)=\bigcup\limits_{r\in R(x)}I(x,r)$.
$\hfill\square$ 
%& \{(x^\prime,e,x),(x^{\prime\prime},e,x)\in \Delta \mid (\forall e\in E)\\
%&(\exists x^\prime,x^{\prime\prime} \in X)  r\subseteq  Reset((x^\prime,e,x))\neq id\\
%& Reset((x^{\prime\prime},e,x))=id, r\subseteq \Gamma ((x^{\prime\prime},e,x))\}.
\end{definition}

The set of output (resp., input) transitions at a timed state $(x,r)$ includes all the transitions that can fire from $x$ (resp. reach $x$) with a clock value in $r$. 
Based on this definition, Algorithm~\ref{alg:zones} merges the regions $r_k, r_{k+1}, \cdots, r_{k^\prime}\in R(x)$ into a zone if the following conditions hold for $i=k+1,\cdots, k^\prime$: (a) $r_i = succ(r_{i-1})$; (b) $I(x,r_k)=I(x,r_i)$, $O(x,r_k)=O(x,r_i)$;
(c) $Reset(\delta)\neq id$ for all $\delta\in I(x,r_i)\cup O(x,r_i)$.
For a region $r\in R(x)$ and a transition $\delta\in I(x,r)\cup O(x,r)$ such that $Reset(\delta)= id$, the region $r$ is included in $Z(x)$ and no merge is done with other regions. 

The set of zones $Z(x)$ partition the clock values $[0,+\infty)$ at $x$ into disjoint segments according to the firability of the input and output transitions.
The maximum number of zones at $x$ is denoted as $Q_x=2(M_x-m_x)+1$.
Note that the partition of zones according to Algorithm~\ref{alg:zones} is not optimal, one can find a more compact partition by further merging the zones associated with $\delta$ where $Reset (\delta) = id$.

%Based on Algorithm~\ref{alg:zones}, one may generate an optimal partition for the clock values by merging the regions associated with $\delta$, where $Reset(\delta)=id$.
%Finally, $\mathcal{O}(x)=\bigcup\limits_{r\in Regions(x)}\mathcal{O}(x,r)$ and $\mathcal{I}(x)=\bigcup\limits_{r\in Regions(x)}\mathcal{I}(x,r)$.
%This leads to Algorithm~\ref{alg:zones} to compute the zones associated with a given discrete state $x\in X$.

\begin{algorithm}
\caption{Computation of the set of zones $Z(x)$}
\label{alg:zones}
%\LinesNumbered
\KwIn{A TFA $G=(X, E, \Delta, \Gamma, Reset, X_0)$, a discrete state $x\in X$}
\KwOut{The set of zones $Z(x)$}
Initialization: 
%$Regions(x)=\{[m,m],\cdots, [M,M]\}$, 
let $m=m_x$, $M=M_x$,
$r=[m,m]$, $z=r$, $Z(x)=\emptyset$\

\While{$r\neq [M,M]$}
{
let $r^{\prime} =succ(r)$
 %and $z_R=z_R\cup \{r\}$

\If{$[O(x,r)=O(x,r^{\prime})]\wedge [I(x,r)=I(x,r^{\prime})] \wedge [(\forall \delta \in O(x,r^{\prime}) \cup I(x,r^{\prime})) Reset(\delta)\neq id]$}
{
let $z=z\cup r^\prime $
}

\Else
{
let $Z(x)=Z(x)\cup \{z\}$ and $z=r^{\prime} $

}
let $r=r^{\prime}$
}
let $Z(x)=Z(x)\cup \{z\}\cup \{(M,+\infty)\}$
\end{algorithm}

\begin{example}
Consider the TFA $G$ in Fig.~\ref{tfaG}.
Algorithm~\ref{alg:zones} provides the set of zones $Z(x_0)=\{[0,0], (0,1), [1,1], (1,3], (3,+\infty)\}$.
In detail, the regions $[0,0], (0,1), [1,1]$ remain in $Z(x_0)$ and do not merge with other regions due to $Reset((x_0,b,x_2))=id$.
The zone $(1,3]$ is obtained by merging the regions $(1,2)$, $[2,2]$, $(2,3)$ and $[3,3]$, which satisfy the \emph{if} condition in line 4 of Algorithm~\ref{alg:zones}.$\hfill\square$
\end{example}

\begin{definition}[\cite{gao2023ifac}]\label{zoneautomaton}
Consider a TFA $G=(X, E, \Delta, \Gamma, X_0)$\footnote{The work \cite{gao2023ifac} assumes that timed automata are associated with a clock that is reset at each event occurrence; consequently, there is no clock resetting function to be considered.} with a single clock that is reset at each event occurrence.
The \emph{zone automaton} of $G$ is an NFA $ZA(G)=(V, E_\tau, \Delta_z, V_0)$, where
\begin{itemize}
  \item $V \subseteq X \times \bigcup\limits_{x\in X}Z(x)$ is the finite set of extended states,
  \item $E_\tau \subseteq E\cup \{\tau\}$ is the alphabet, where the event $\tau$ implies time elapsing from any clock value $\theta\in z$ to any $\theta^\prime\in succ(z)$ when $G$ stays at $x\in X$;
  \item $\Delta_z\subseteq V \times E_\tau \times V$ is the transition relation,
  \item $V_0=\{(x,[0,0])\mid x\in X_0\}\subseteq V$ is the set of initial extended states.$\hfill\square$
\end{itemize}
\end{definition}

The work in \cite{gao2023ifac}  corresponds to the particular case with
$Reset(\delta)=[0,0]$ for each transition $\delta\in\Delta$ and a maximal dwell time for each discrete state.
The zone automaton is a finite state automaton that provides a purely discrete event description of the behaviour of a given timed automaton.
Each state of a zone automaton is called an \emph{extended state}, which is a pair $(x,z)$ whose first element is a discrete state $x\in X$ and whose second element is a zone $z\in Z(x)$ specifying the range of the clock values.
The extended state evolves either because of the elapsed time $\tau$ or because of the occurrence of a discrete event.
%The first case is a time-driven evolution, and the second case is an event-driven evolution.
In the former, a transition $((x,z),\tau,(x,succ(z)))\in \Delta_z$ corresponds to a time-driven evolution of $G$ from a clock value in $z$ to another clock value in $succ(z)\in Z(x)$.
In the latter, an event-driven evolution caused by an event $e\in E$ from $(x,z)$ leads to $(x^\prime, z^\prime)$ indicating that the occurrence of event $e$ yields $x^\prime$ with a clock value in $z^\prime$.

In this paper, the zone automaton associated with a TFA can be constructed by implementing the event-driven evolution according to both the timing function and the clock resetting function.
%A transition $((x,z),e,(x^\prime,z^\prime))\in \Delta_z$ is defined if $z\in Z(x)$, $z^\prime\in Z(x^\prime)$, $(x,e,x^\prime)\in\mathcal{O}(x,z)$, and $(x,e,x^\prime)\in\mathcal{I}(x^\prime,z^\prime)$.
%The zone automaton associated with a TFA $G=(X, E, \Delta, \Gamma, Reset, X_0)$ can be constructed by implementing the event-driven evolution according to the clock resetting function.
In detail, for each transition $(x,e,x^\prime)\in\Delta$, we first determine whether $z\subseteq\Gamma((x,e,x^\prime))$; if so, for each zone $z^\prime\in Z(x^\prime)$, a transition $((x,z), e, (x^\prime, z^\prime))$ is defined if (a) $Reset((x,e,x^\prime))\neq id$ and $z^\prime \subseteq Reset((x,e,x^\prime))$ (the clock is reset after $(x,e,x^\prime)$ occurs), or if (b) $Reset((x,e,x^\prime))= id$ and $z=z^\prime$ (the clock is not reset after $(x,e,x^\prime)$ occurs).
We further define the function $f_x: V \rightarrow X$ (resp., $f_z: V \rightarrow \bigcup\limits_{x\in X} Z(x)$) mapping an extended state in $V$ to a discrete state in $X$ (resp., a zone of the associated discrete state).

\begin{example}
Consider the TFA $G=(X$, $E$, $\Delta$, $\Gamma$, $Reset$, $X_0)$ in Fig.~\ref{tfaG}.
The zone automaton $G_z=(V, E_\tau, \Delta_z, V_0)$ is shown in Fig.~\ref{lambdaautomaton}.
For instance, transition $((x_0,[0,0]), \tau, (x_0,(0,1)))$ implies that the clock may evolve from $0$ to any value in $(0,1)$ if $G$ is at $x_0$.
Three transitions labeled with $b$ go from $(x_0,[0,0])$ to $(x_2,[0,0])$, from $(x_0,(0,1))$ to $(x_2,(0,1))$, and from $(x_0,[1,1])$ to $(x_2,[1,1])$, respectively. Each transition implies an event-driven evolution from $x_0$ to $x_2$ under the occurrence of a transition $(x_0,b,x_2)$.$\hfill \square$
%Note that the pair $(x_1,[0,1))$ is not an extended state in $V$ because it cannot be reached by any transition.\hfill $\square$
\end{example}

\begin{figure}[h]
\centering
\includegraphics[width=10cm]{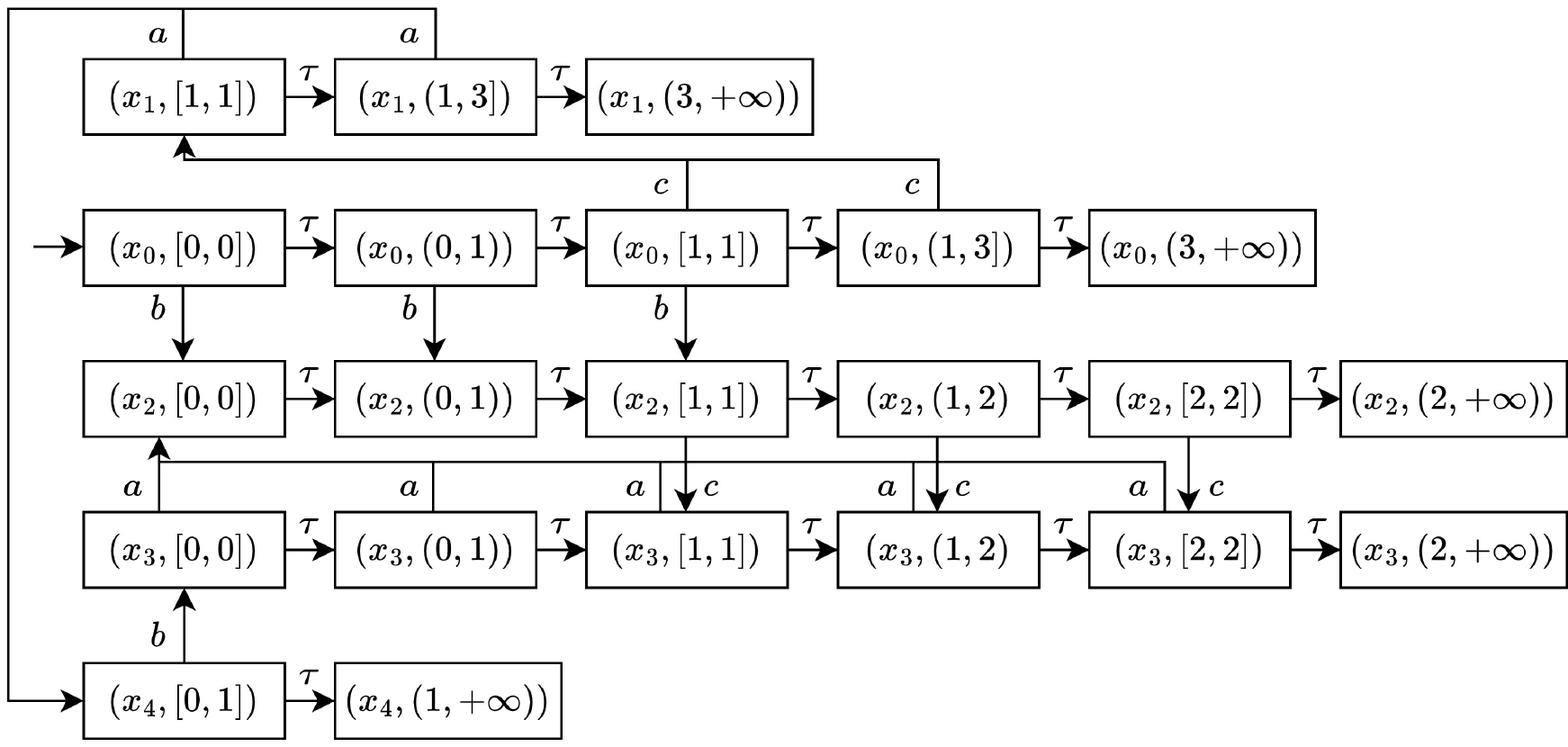}
\caption{Zone automaton $G_z$ of the TFA $G$ in Fig.~\ref{tfaG}.}\label{lambdaautomaton}
\end{figure}

\section{Dynamics of a zone automaton}

In this section, we explore the dynamics of a zone automaton $G_z=(V, E_\tau, \Delta_z$, $V_0)$ associated with a TFA $G=(X, E, \Delta, \Gamma, Reset, X_0)$ and discuss how the timed evolutions of $G$ are related to the evolutions of zone automaton $G_z$.
%We first introduce the notion of \emph{$\tau$-run at a discrete state $x\in X$} that corresponds to the elapse of time with no event occurrence.
%Next, a run in $G_z$ is considered, including also the occurrence of discrete events, to show the time-driven evolution and the event-driven evolution of $G$.
%We show that the problem of investigating the reachability of a discrete state in $G$ can be reduced to the reachability analysis of an extended state in $G_z$.

\begin{definition}
Given a zone automaton $G_z=(V$, $E_\tau$, $\Delta_z$, $V_0)$ and a discrete state $x$ of a TFA $G$, a \emph{$\tau$-run at $x$} of length $k\ (1 \leq k\leq n)$ is defined as a sequence of $k$ extended states $(x,z_{(i)})\in V\ (i= 1,\cdots,k)$, and event $\tau\in E_\tau$, represented as
$\rho_\tau(x): (x,z_{(1)})\stackrel{\tau}{\longrightarrow}\cdots \stackrel{\tau}{\longrightarrow}(x,z_{(k)})$
such that $((x,z_{(i)}),\tau,(x,z_{(i+1)}))\in \Delta_z$ holds for $i\in\{1,\cdots,k-1\}$.
The \emph{starting extended state} and the \emph{ending extended state} of $\rho_\tau(x)$ are denoted as $v_{st}(\rho_\tau(x))=(x,z_{(1)})$ and $v_{en}(\rho_\tau(x))=(x,z_{(k)})$, respectively.
The \emph{duration range} of $\rho_\tau(x)$ is the time distance of $z_{(1)}$ and $z_{(k)}$, denoted as $d(\rho_\tau(x))= D(z_{(1)},z_{(k)})$.$\hfill \square$
\end{definition}

\begin{definition}
Given a zone automaton $G_z=(V$, $E_\tau$, $\Delta_z$, $V_0)$ of a TFA $G$, a run of length $k\geq 0$ is a sequence of $k+1$ $\tau$-runs $\rho_{\tau}(x_{(i)})\ (i= 0,\cdots,k)$ at $x_{(i)}\in X$, and $k$ events $e_i\in E\ (i= 1,\cdots,k)$, represented as
$\bar{\rho}:\rho_{\tau}(x_{(0)}) \stackrel{e_{1}}{\longrightarrow}\rho_{\tau}(x_{(1)})\cdots \stackrel{e_{k}}{\longrightarrow} \rho_{\tau}(x_{(k)}),$
such that $(v_{en}(\rho_\tau(x_{(i-1)})),e_i, v_{st}(\rho_\tau(x_{(i)})))\in \Delta_z$ holds for $i\in\{1,\cdots,k\}$.
In addition, the \emph{starting extended state} and the \emph{ending extended state} of $\bar{\rho}$ are defined as $v_{st}(\bar{\rho})=v_{st}(\rho_{\tau}(x_{(0)}))$ and $v_{en}(\bar{\rho})=v_{en}(\rho_{\tau}(x_{(k)}))$, respectively.
The \emph{duration range} of $\bar{\rho}$ is defined as $d(\bar{\rho})=\bigoplus\limits_{i=1}^k d(\rho_\tau(x_i))$.
The logical word generated by $\bar{\rho}$ is denoted as $s(\bar{\rho})=e_1\cdots e_k$ via a function defined as $s:E_\tau^\ast \rightarrow E^\ast$. The set of runs generated by $G_z$ is defined as $\mathcal{R}_z(G_z)$.\hfill $\square$
\end{definition}

%{\blue
%A $\tau$-run at $x$ involves a series of evolutions of extended states in $G_z$ associated with $x$ caused by the event $\tau$ and implies the time elapsing discretely while $G$ is at $x$..
%The dynamics of a zone automaton $G_z$ can be represented by its runs, each of which is a sequence of $\tau$-runs connected by an evolution caused by an event $e_i\in E$.
%The logical word of $\bar{\rho}$ is the sequence of events in $E$ that have been involved in $\bar{\rho}$.
%The duration range of $\bar{\rho}$ is evaluated by summing up the duration ranges of all the involved $\tau$-runs that describes the possible duration of the run.
%}

%Recall that a discrete state $x^\prime\in X$ is $T$-reachable from $x\in X$ if there exists a timed evolution $(\sigma(\rho),t)$ such that $x_{st}(\rho)=x$, $x_{en}(\rho)=x^\prime$, and $t-t_{st}(\rho)=T$.
The dynamics of a zone automaton $G_z$ can be represented by its runs, each of which is a sequence of $\tau$-runs, which implies the time elapses discretely at discrete states, connected by an evolution caused by an event $e_i\in E$.
We now provide a sufficient and necessary condition for the reachability of a discrete state in a TFA and explain the correlation of the dynamics of a TFA and that of zone automaton.

\begin{theorem}
Given a TFA $G=(X, E, \Delta, \Gamma, Reset, X_0)$, zone automaton $G_z=(V, E_\tau, \Delta_z, V_0)$ and a time instant $T\in\mathbb{R}_{\geq 0}$, $x^\prime\in X$ is $T$-reachable from $x\in X$ if and only if there exists a run $\bar{\rho}$ in $G_z$ such that $T\in d(\bar{\rho})$, $v_{st}(\bar{\rho})=(x,z)$, and $v_{en}(\bar{\rho})=(x^\prime,z^\prime)$, where $z\in Z(x)$ and $z^\prime\in Z(x^\prime)$.
\end{theorem}
\begin{proof}
\textit{(if)} Let $x = x_{(0)}$, $x^\prime = x_{(k)}$, $t_0$, $t_1, \cdots, t_k\in \mathbb{R}_{\geq 0}$ and $\bar{\rho}: \rho_\tau(x_{(0)}){\xrightarrow{e_1}}\cdots {\xrightarrow{e_k}}$ $\rho_\tau(x_{(k)})$ be a run such that $e_i\in E$, $t_{i}-t_{i-1}\in d(\rho_\tau(x_{(i-1)}))$ for $i\in \{1,\cdots, k\}$, $t-t_k \in d(\rho_\tau(x_{(k)}))$, and $T=t-t_0\in d(\bar{\rho})$.
It can be inferred that $(v_{en}(\rho_\tau(x_{(i-1)})),e_i,v_{st}(\rho_\tau(x_{(i)})))\in \Delta_z$ holds for $i\in \{1,\cdots, k\}$.
Accordingly, for $i\in \{1,\cdots, k\}$, there exist $\theta_{(i-1)}+t_i-t_{i-1}\in f_z(v_{en}(\rho_\tau(x_{(i-1)})))$ and
$\theta_{(i)}\in f_z(v_{st}(\rho_\tau(x_{(i)})))$ such that $(x_{(i-1)},e_i,x_{(i)})\in O(x_{(i-1)},f_z(v_{st}(\rho_\tau(x_{(i-1)}))))$ and $(x_{(i-1)},e_i,x_{(i)})\in I(x_{(i)},f_z(v_{st}(\rho_\tau(x_{(i)}))))$.
It is obvious that there exists a timed evolution $(\sigma(\rho),t)$, where the timed run $\rho$:$(x_{(0)},\theta_{(0)}){\xrightarrow{(e_1, t_1)}}$ $\cdots$ $ {\xrightarrow{(e_{k}, t_k)}}(x_{(k)}$, $\theta_{(k)})$ satisfies the condition $T(\rho)= t_k-t_0$.
Thus, $x^\prime$ is $T$-reachable from $x$.

\textit{(only if)} Let $(\sigma(\rho),t)\in (E\times \mathbb{R}_{\geq 0})^\ast \times \mathbb{R}_{\geq 0}$ be a timed evolution of $G$ from $t_0\in\mathbb{R}_{\geq 0}$ to $t\in\mathbb{R}_{\geq 0}$ such that $x_{st}(\rho)=x$ and $x_{en}(\rho)= x^\prime$.
The proof is made by induction on the length $k$ of the timed run $\rho$ generated by $G$ from $t_0$ to $t$.
The base case is for the timed run $\rho$ of length 0 that involves only the discrete state $x$ and no transition in $G$. It is $\sigma(\rho)=\lambda$ and $x_{st}(\rho)=x_{en}(\rho)=x=x^\prime$.
There exists a run $\bar{\rho}:(x,\bar{z})\stackrel{\tau}{\longrightarrow}\cdots\stackrel{\tau}{\longrightarrow}(x,z)$ in $G_z$, where $T=t-t_0\in d(\bar{\rho})$.
Thus the base case holds.
By denoting $x = x_{(0)}$ and $x^\prime = x_{(k)}$, the induction hypothesis is that the existence of a timed evolution $(\sigma(\rho),t)$ generating from $t_0$, where $\rho: (x_{(0)},\theta_{(0)}){\xrightarrow{(e_1,t_1)}}\cdots {\xrightarrow{(e_k,t_k)}}(x_{(k)},\theta_{(k)})$ of length $k\geq 1$ with $x_{(i)}\in X$ and $e_i\in E$ for all $i\in \{1,\cdots, k\}$, implies the existence of a run $\bar{\rho}:\rho_\tau(x_{(0)}){\xrightarrow{e_1}}\cdots {\xrightarrow{e_k}}\rho_\tau(x_{(k)})$ in $G_z$ such that $t- t_0\in d(\bar{\rho})$, $t-t_k\in d(\rho_\tau(x_{(k)}))$, $t_i-t_{i-1}\in d(\rho_\tau(x_{(i-1)}))$ for $i \in\{1,\cdots,k\}$, $v_{st}(\bar{\rho})=(x_{(0)},z)$ and $v_{en}(\bar{\rho})=(x_{(k)},\bar{z})$, where $z\in Z(x_{(0)})$, and $\bar{z}\in Z(x_{(k)})$.
We now prove that the same implication holds for a timed evolution $(\sigma(\rho^\prime),t^\prime)$, where $\rho^\prime$: $\rho{\xrightarrow{(e_{k+1},t)}}(x_{(k+1)},\theta_{(k+1)})$.
According to $\rho^\prime$, it is $\theta_{(k+1)}\in Reset((x_{(k)}$, $e_{k+1}$, $x_{(k+1)}))$ if $Reset((x_{(k)},e_{k+1},x_{(k+1)}))\neq id$, or $\theta_{(k+1)}\in \Gamma((x_{(k)},e_{k+1},x_{(k+1)}))$ if $Reset((x_{(k)},e_{k+1},x_{(k+1)}))= id$.
In addition, $\theta_{(k)}+t-t_k\in \Gamma((x_{(k)},e_{k+1},x_{(k+1)}))$.
It implies that %$(x_{(k)},e_{k+1},x_{(k+1)})\in\mathcal{O}(x_{(k)}, \theta_{(k)}+t-t_k)$ and $(x_{(k)},e_{k+1},x_{(k+1)})\in\mathcal{I}(x_{(k+1)}, \theta_{(k+1)})$; consequently, 
there exists a run $\bar{\rho}^\prime:\bar{\rho}{\xrightarrow{e_{k+1}}}\rho_\tau(x_{({k+1})})$ in $G_z$ such that $\theta_{(k+1)}\in f_z(v_{st}(\rho_\tau(x_{({k+1})})))$ and $\theta_{(k+1)}+t^\prime-t\in f_z(v_{en}(\rho_\tau(x_{({k+1})})))$.
Therefore, it is $t^\prime-t_0\in d(\bar{\rho}^\prime)$ according to $t^\prime-t\in d(\rho_\tau(x_{({k+1})})))$ and $t-t_0\in d(\bar{\rho})$.
\end{proof}

%Theorem 1 provides a necessary and sufficient condition to determine the $T$-reachability of a discrete state $x^\prime\in X$ from $x\in X$, where $T\in\mathbb{R}_{\geq 0}$.
According to Theorem 1, if $x^\prime$ is $T$-reachable from $x$, then in the zone automaton $G_z$ there exists a run that originates from $(x, z)$ and reaches $(x^\prime,z^\prime)$, where $z\in Z(x)$ and $z^\prime\in Z(x^\prime)$.
In addition, $T$ belongs to the duration range of that run.
In turn, given a run $\bar{\rho}$ in $G_z$ such that $T\in d(\bar{\rho})$, it can be concluded that the discrete state associated with $v_{en}(\bar{\rho})$ is $T$-reachable from the discrete state associated with $v_{st}(\bar{\rho})$.
In simple words, the $T$-reachability of $x^\prime$ from $x$ can be analyzed by exploring an appropriate run in the zone automaton.

\begin{example}
Consider the TFA $G$ in Fig.~\ref{tfaG} and zone automaton $G_z$ in Fig.~\ref{lambdaautomaton}.
A timed evolution $((c,1.5)(a,3),4)$ of $G$ from $0$ to $4$ implies that $x_4$ is $4$-reachable from $x_0$, which can be concluded from a run in $G_z$ as $\bar{\rho}:\rho_{\tau}(x_0)\stackrel{c}{\longrightarrow}\rho_{\tau}(x_1)\stackrel{a}{\longrightarrow}\rho_{\tau}(x_4)$, where $\rho_{\tau}(x_0):(x_0,[0,0]))\stackrel{\tau}{\longrightarrow}(x_0,(0,1))\stackrel{\tau}{\longrightarrow}$$(x_0,[1, 1])$$\stackrel{\tau}{\longrightarrow}(x_0,(1,3])$, $\rho_{\tau}(x_1): (x_1,[1,1])\stackrel{\tau}{\longrightarrow}(x_1,(1,3])$ and $\rho_{\tau}(x_4): (x_4,[0,1]).$
\hfill$\square$
\end{example}

\section{State estimation of TFA}

Given a partially observed TFA $G=(X$, $E$, $\Delta$, $\Gamma$, $Reset$, $X_0)$ with a partition of alphabet $E=E_o\dot\cup E_{uo}$, in this section we develop an approach for state estimation based on the zone automaton $G_z$, given a timed observation $(\sigma_o,t)\in (E_o \times \mathbb{R}_{\geq 0})^\ast \times \mathbb{R}_{\geq 0}$.

We partition this section into two subsections.
In the first subsection, we consider the case where $G$ produces no observation, which is an intermediate step towards the solution of the state estimation problem under partial observation.
In the second subsection, we take into account the information coming from the observation of new events at certain time instants, and prove that the discrete states consistent with a timed observation $(\sigma_o,t)$ and the range of clock value associated with each estimated discrete state can be inferred following a finite number of runs in the zone automaton $G_z$.
We make the following assumption.
\begin{assumption}[RO: Reinitialized observations]
The clock is reset upon the occurrence of any observable event, i.e.,
$$(x,e,x') \in \Delta, \ e\in E_o \quad \Longrightarrow \quad Reset(x,e,x') \neq id.$$ $\hfill\square$
%\null \hfill $\square$
\end{assumption}

This assumption is necessary to ensure that the defined zone automaton contains all relevant information to estimate the discrete state.
Consider a scenario where an observable event occurs without resetting the clock. In such a case, for future estimations one may need to keep track of this exact value adding new extended states to the zone automaton: thus, the state space of the zone automaton could grow indefinitely as new events are observed.

\subsection{State estimation under no observation}
%Given an extended state of a zone automaton, which refers to a discrete state of the associated TFA and a range of the possible clock value, this subsection provides a criterion to estimate the set of extended states in which the zone automaton can be when no observation is received for a given time instant $t\in\mathbb{R}_{\geq 0}$.
%To formalize this we preliminarily propose the following definition.

\begin{definition}
Given a TFA $G=(X$, $E$, $\Delta$, $\Gamma$, $Reset$, $X_0)$ with set of unobservable events $E_{uo}$ and zone automaton $G_z=(V, E_\tau, \Delta_z, V_0)$, the following set of extended states
$V_\lambda(x, z, t)= \{v_{en}(\bar{\rho})\in V \mid   (\exists \bar{\rho}\in \mathcal{R}_z (G_z)) t \in d(\bar{\rho}),  v_{st}(\bar{\rho})=(x,z),  s(\bar{\rho})\in E_{uo}^\ast \}$
is said to be \emph{$\lambda$-estimation from $(x,z)\in V$ within $t\in\mathbb{R}_{\geq 0}$}.
\hfill $\square$
\end{definition}

Given a zone automaton, the $\lambda$-estimation from $(x,z)\in V$ within $t\in\mathbb{R}_{\geq 0}$ is the set of extended states of $G_z$ that can be reached following a run of duration $t$, originating at $(x, z)$ and producing no observation.
If $(x^\prime,z^\prime)\in V_\lambda(x, z, t)$, it basically reveals that $x^\prime$ is unobservably $T$-reachable from $x$ with a clock value $\theta\in z$ according to Theorem 1.
In addition, the zone $z^\prime$ associated with $x^\prime$ specifies the range of the value of the clock.
The $\lambda$-estimation from $(x,z)$ within $t$ can be obtained by enumerating the runs starting from $(x,z)$ with a duration associated with $t$.
By denoting the maximum number of zones for $x\in X$ as $Q_x$, the complexity for computing the $\lambda$-estimation is $\mathcal{O}(q^3|X|)$ with $q=\max\limits_{x\in X} |Q_x|$.

%According to Theorem 1, the $\lambda$-estimation basically reveals the unobservably $T$-reachable discrete states from $x$ with a clock value $\theta\in z$, i.e., the set of discrete states of $G$ when it starts its timed evolution from $(x,\theta)$ with $\theta\in z$ and produces no observation for time $T$.
%If $(x^\prime,z^\prime)\in V_\lambda(x, z, t)$, then the zone $z^\prime$ associated with $x^\prime$ specifies the range of the value of the clock.

\begin{proposition}
Given a TFA with a set of discrete states $X$, the set of unobservable events $E_{uo}$, and zone automaton $G_z=(V, E_\tau, \Delta_z, V_0)$,
$x^\prime\in X$ is unobservably $T$-reachable from $x\in X$, where $T\in \mathbb{R}_{\geq 0}$, if and only if there exist $z\in Z(x)$ and $z^\prime\in Z(x^\prime)$ such that $(x^\prime,z^\prime)\in V_\lambda(x,z,T)$.
\end{proposition}
\begin{proof}
\textit{(if)} Suppose that there exist $z\in Z(x)$ and $z^\prime\in Z(x^\prime)$ such that $(x^\prime,z^\prime)\in V_\lambda(x,z,T)$.
There exists a run $\bar{\rho}$ in $G_z$ such that $v_{st}(\bar{\rho})=(x,z)$, $v_{en}(\bar{\rho})=(x^\prime,z^\prime)$, $s(\bar{\rho})\in E_{uo}^\ast$ and $T\in d(\bar{\rho})$.
According to Theorem 1, $x^\prime$ is unobservably $T$-reachable from $x$.

\textit{(only if)} Let $x^\prime$ be unobservably $T$-reachable from $x$. Then, there exists a timed evolution $(\sigma(\rho),t)$ from $t_0$ such that $x_{st}(\rho)=x$, $x_{en}(\rho)=x^\prime$, $T=t-t_0$, and $s(\bar{\rho})\in E_{uo}^\ast$.
Accordingly, there exists a run $\bar{\rho}$ in $G_z$ such that $v_{st}(\bar{\rho})=(x,z)$, $v_{en}(\bar{\rho})=(x^\prime,z^\prime)$, $s(\bar{\rho})\in E_{uo}^\ast$ and $T \in d(\bar{\rho})$, where $z\in Z(x)$ and $z^\prime\in Z(x^\prime)$.
Thus, $(x^\prime,z^\prime)\in V_\lambda(x,z,T)$.
\end{proof}

\begin{table*}[t]
\caption{State estimation of the TFA $G$ in Fig. 1 under no observation for $t\in[0,2]$.}
\begin{center}\scalebox{0.75}{
\begin{tabular}{c|c|c|c}
\toprule[1pt]
$k$ & Time interval $I_k$ &  $\lambda$-estimation $V_\lambda(x_0,[0,0],t)$, where $t\in I_k$ & \tabincell{c}{ $\mathcal{X}(\lambda, t), t\in I_k$ }\\%& the set of enabled observable events
\hline
$0$ & [0,0] & $\{(x_0,[0,0]), (x_2,[0,0])\}$ & $\{x_0, x_2\}$ \\%& $\emptyset$\\
$1$ & (0,1) & $\{(x_0,(0,1)), (x_2,(0,1))\}$& $\{x_0, x_2\}$ \\%& $\emptyset$\\
$2$ & [1,1] & $\{(x_0,[1,1]),(x_1,[1,1]),(x_2,[1,1]),(x_3,[1,1])\}$& $\{x_0, x_1, x_2, x_3\}$ \\%& $\{a\}$\\
$3$ & (1,2) & $\{(x_0,(1,3]),(x_1,[1,1]),(x_1,(1,3]),(x_2,(1,2)),(x_3,(1,2))\}$& $\{x_0, x_1, x_2, x_3\}$ \\
$4$ & [2,2] & $\{(x_0,(1,3]),(x_1,[1,1]),(x_1,(1,3]),(x_2,[2,2]),(x_3,[2,2])\}$& $\{x_0, x_1, x_2, x_3\}$ \\
%$5$ & (2,3) & {\blue $(x_0,(1,3]),(x_1,[1,1]),(x_1,(1,3]),(x_2,(2,+\infty)),(x_3,(2,+\infty))$}& $\{x_0, x_1, x_2, x_3\}$ \\
%$6$ & [3,3] & {\blue $(x_0,(1,3]),(x_1,[1,1]),(x_1,(1,3]),(x_2,(2,+\infty)),(x_3,(2,+\infty))$}& $\{x_0, x_1, x_2,x_3\}$ \\
\bottomrule[1pt]
\end{tabular}}
\end{center}\label{lambdaestimation}
\end{table*}
\begin{example}
Consider the TFA $G$ in Fig.~\ref{tfaG}, where $E_o=\{a\}$, $E_{uo}=\{b,c\}$, and zone automaton $G_z$ in Fig.~\ref{lambdaautomaton}.
Given the following runs in $G_z$ of duration 1 starting from $(x_0,[0,0])$ and involving no observable events: (1) $\bar{\rho}_1: \rho_{\tau}(x_0)$; (2) $\bar{\rho}_2: \rho_{\tau}(x_0)\stackrel{c}{\longrightarrow}\rho_{\tau}(x_1)$; (3) $\bar{\rho}_3: \rho_{\tau}(x_0)\stackrel{b}{\longrightarrow}\rho_{\tau}(x_2)$; and (4) $\bar{\rho}_4: \rho_{\tau}(x_0)\stackrel{b}{\longrightarrow}\rho_{\tau}(x_2)\stackrel{c}{\longrightarrow}\rho_{\tau}(x_3)$,
where $\rho_{\tau}(x_0):(x_0,[0,0])\stackrel{\tau}{\longrightarrow}(x_0,(0,1))\stackrel{\tau}{\longrightarrow}(x_0,[1,1])$, $\rho_{\tau}(x_1): (x_1,[1,1])$, $\rho_{\tau}(x_2): (x_2,[1,1])$, and $\rho_{\tau}(x_3): (x_3,[1,1])$, it can be inferred that $V_\lambda(x_0,[0,0], 1)= \{(x_0,[1,1]), (x_1,[1,1]), (x_2,[1,1]), (x_3,[1,1])\}$.
%It is $V_\lambda(x_0,[0,0], 1)= \{(x_0,[1,1]), (x_1,[1,1]), (x_2,[1,1]), (x_3,[1,1])\}$, which results from the following runs in $G_z$ of duration 1 starting from $(x_0,[0,0])$ and involving no observable events: (1) $\bar{\rho}_1: \rho_{\tau}(x_0)$; (2) $\bar{\rho}_2: \rho_{\tau}(x_0)\stackrel{c}{\longrightarrow}\rho_{\tau}(x_1)$; (3) $\bar{\rho}_3: \rho_{\tau}(x_0)\stackrel{b}{\longrightarrow}\rho_{\tau}(x_2)$; and (4) $\bar{\rho}_4: \rho_{\tau}(x_0)\stackrel{b}{\longrightarrow}\rho_{\tau}(x_2)\stackrel{c}{\longrightarrow}\rho_{\tau}(x_3)$,
%where $\rho_{\tau}(x_0):(x_0,[0,0])\stackrel{\tau}{\longrightarrow}(x_0,(0,1))\stackrel{\tau}{\longrightarrow}(x_0,[1,1])$, $\rho_{\tau}(x_1): (x_1,[1,1])$, $\rho_{\tau}(x_2): (x_2,[1,1])$, and $\rho_{\tau}(x_3): (x_3,[1,1])$.
%For instance, $(x_3,[1,1])\in V_\lambda(x_0,[0,0], 1)$ implies that $G$ can be at $x_3$ with clock value in the zone $[1,1]$ if $G$ starts its evolution from $x_0$ produces no observation during the whole time interval $[0,1]$.
Table \ref{lambdaestimation} summarizes the $\lambda$-estimation and the set of discrete states consistent with $(\lambda,t)$, where $t$ belongs to a region of $[0,2]$.$\hfill\square$
\end{example}

\subsection{State estimation under partial observation}
In this subsection we focus on the most general state estimation problem when a timed observation is received as a pair of a non-empty timed word and a time instant.
We first propose a general result that characterizes the set of discrete states of a TFA consistent with a given timed observation, by means of the extended states reachable in zone automaton.

\begin{theorem}\label{observationconsistent}
Consider a TFA $G=(X$, $E$, $\Delta$, $\Gamma$, $Reset$, $X_0)$ with set of observable events $E_o$.
Given a timed observation $(\sigma_o,t)\in (E_o \times \mathbb{R}_{\geq 0})^\ast \times \mathbb{R}_{\geq 0}$, it is $x\in \mathcal{X}(\sigma_o,t)$ if and only if there exists a run $\bar{\rho}$ in the zone automaton $G_z$ such that $f_x(v_{st}(\bar{\rho}))\in X_0$, $f_x(v_{en}(\bar{\rho}))=x$, $t\in d(\bar{\rho})$ and $P_l(s(\bar{\rho}))=S(\sigma_o)$, where $P_l: E^\ast \rightarrow E_o^\ast$.
\end{theorem}

\begin{proof}
\textit{(if)}
In the case that no observation is contained in $\sigma_o$, let $\bar{\rho}: \rho_\tau(x_{0})$ be a run in $G_z$, where $x_0\in X_0$. We have $f_x(v_{en}(\bar{\rho}))=x_0$, $t\in d(\bar{\rho})$ and $\sigma_o=\lambda$.
In this case there exists only one timed evolution $(\lambda,t)\in \mathcal{E}(G,t)$ such that $x_0\in \mathcal{X}(\lambda,t)$.
In the case that there exist one or more event occurrences, let us suppose that $x_{(0)}\in X_0$, $x_{(k)}=x$, and a run in $G_z$ as $\bar{\rho}: \rho_\tau(x_{(0)}){\xrightarrow{e_{1}}}\cdots {\xrightarrow{e_{k}}}\rho_\tau(x_{(k)})$ such that $t\in d(\bar{\rho})$, $x_{(i)}\in X$ and $e_{i}\in E$ for $i\in\{0,\cdots,k\}$.
It can be inferred that there exists a timed evolution $(\sigma(\rho),t)$ such that $x_{(k)}$ is $t$-reachable from $x_{(0)}$, where the timed run $\rho:(x_{(0)},\theta_{(0)}){\xrightarrow{(e_1,t_1)}}\cdots {\xrightarrow{(e_{k},t_k)}}(x_{(k)},\theta_{(k)})$ satisfies for $i\in\{1,\cdots, k\}$ that $t-t_k\in d(\rho_\tau(x_{(k)}))$, $t_i-t_{i-1}\in d(\rho_\tau(x_{(i-1)}))$, $\theta_{(i-1)}+t_i-t_{i-1}\in \Gamma((x_{(i-1)},e_{i},x_{(i)}))$ , $\theta_{(i)}\in Reset((x_{(i-1)},e_{i},x_{(i)}))$ if $Reset((x_{(i-1)},e_{i},x_{(i)}))\neq id$, and $\theta_{(i)}\in \Gamma((x_{(i-1)},e_{i},x_{(i)}))$  if $Reset((x_{(i-1)},e_{i},x_{(i)}))= id$.
According to $P_l(\bar{s}(\bar{\rho}))=S(\sigma_o)$, we have $(\sigma(\rho),t)\in \mathcal{S}(\sigma_o,t)$. Obviously, $x\in \mathcal{X}(\sigma_o,t)$ holds.

\textit{(only if)} In the case that $\sigma_o= \lambda$, $x\in \mathcal{X}(\lambda,t)$ holds. There exists a run $\bar{\rho}$ in $G_z$ such that $x$ is unobservably $t$-reachable from $x_0\in X_0$. Consequently, $f_x(v_{en}(\bar{\rho}))=x$, $t\in d(\bar{\rho})$ and $\bar{s}(\bar{\rho})=\varepsilon$ hold.
In the case that $\sigma_o=(e_{o1},t_1)\cdots(e_{ok},t_k)$, where $k\geq 1$, $0\leq t_1 \leq\cdots\leq t_k \leq t$ and $e_{oi}\in E_o \ (i\in\{1,\cdots, k\})$, it is $x\in \mathcal{X}(\sigma_o,t)$, implying that $x$ is $t$-reachable from $x_0\in X_0$. Consequently, there exists a run $\bar{\rho}$ in $G_z$ such that $t\in d(\bar{\rho})$, $f_x(v_{st}(\bar{\rho}))=x_0$ and $f_x(v_{en}(\bar{\rho}))=x$.
Let $\bar{\rho}$ be a sequence of $k$ runs $\bar{\rho}_i$ and $k$ observable events $e_{oi}$, where $i\in \{1,\cdots, k\}$, as $\bar{\rho}: \bar{\rho}_0{\xrightarrow{e_{o1}}}\cdots {\xrightarrow{e_{ok}}}\bar{\rho}_k$. Obviously, $\bar{s}(\bar{\rho})=e_{o1}\cdots e_{ok}$, $t-t_k \in d(\bar{\rho}_k)$ and $t_i - t_{i-1} \in d(\bar{\rho}_{i-1})$ hold for each $i\in \{1,\cdots, k\}$.
\end{proof}

In simple words, $x\in \mathcal{X}(\sigma_o,t)$ if there exists a run $\bar{\rho}$ in $G_z$ that produces the same logical observation as $\sigma_o$ from a discrete state in $X_0$ at $0$ and reaches $x$ at $t$. In turn, $x\in \mathcal{X}(\sigma_o,t)$ can be computed following run $\bar{\rho}$ in $G_z$.
%In more detail, the run $\bar{\rho}$ contains no observable event if $\sigma_o=\lambda$. This naturally implies that $x\in \mathcal{X}(\lambda,t)$.
%If $\sigma_o\neq \lambda$, the run $\bar{\rho}$ is a sequence of $k+1$ runs $\bar{\rho}_i,\ i\in\{0,\cdots,k\}$, and $k$ observable events $e_{oi}\in E_o,\ i\in\{1,\cdots,k\}$, as $\bar{\rho}: \bar{\rho}_0{\xrightarrow{e_{o1}}}\cdots {\xrightarrow{e_{ok}}}\bar{\rho}_k$. For each $i\in \{0,\cdots,k\}$, we have $s(\bar{\rho}_i)=\varepsilon$, implying that $\bar{\rho}_i$ produces no observation. For each $i\in \{1,\cdots,k\}$, an observable event $e_{oi}\in E_o$ leads the extended state evolution from $v_{en}(\bar{\rho}_{i-1})$ to $v_{st}(\bar{\rho}_i)$.

\begin{proposition}
Consider a TFA $G$ with set of observable events $E_{o}$ that produces two timed observations $(\sigma_{o},t_i),(\sigma_{o},t)\in (E_o \times \mathbb{R}_{\geq 0})^\ast \times \mathbb{R}_{\geq 0}$, where $\sigma_{o}=(e_{o1},t_1)\cdots(e_{oi},t_i)$ and $t_1\leq\cdots\leq t_i \leq t$ for $i\geq 1$.
For each timed state $(x,\theta)$ reached by a timed evolution in $\mathcal{S}(\sigma_{o},t_i)$, and for each $v\in V_\lambda(x,z,t-t_i)$, where $z\in Z(x)$ and $\theta\in z$, it holds $f_x(v)\in \mathcal{X}(\sigma_{o},t)$ if Assumption RO holds.
\end{proposition}

\begin{proof}
Let a timed run $\rho:(x_0,0)\xrightarrow{(e_1,t_1)}(x_{(1)},\theta_{(1)})\cdots \xrightarrow{(e_k,t_k)}(x_{(k)},\theta_{(k)})\xrightarrow{(e_{oi},t_{i})}(x,\theta)$, where $x_{(1)},\cdots,x_{(k)}\in X$, $e_1,\cdots, e_k\in E$, and $P(\sigma(\rho))=\sigma_o$, the timed state $(x,\theta)$ is reached by the timed evolution $(\sigma(\rho),t_i)\in\mathcal{S}(\sigma_{o},t_i)$.
According to Theorem 1, it can be inferred that there exists a run in $G_z$ as $\bar{\rho}:(x_0,[0,0])\rightarrow\cdots\rightarrow (x,z)$, where $\theta\in z$.
If Assumption RO holds, it implies that there exists $z\in Z(x)$ such that $\theta\in z$ and $z\subseteq Reset ((x_{(k)},e_{oi},x))$.
Given $v\in V_\lambda(x,z,t-t_i)$, there exists a run $\bar{\rho}^\prime:(x,z)\rightarrow \cdots \rightarrow v$ such that $P_l(s(\bar{\rho}^\prime))=\varepsilon$ and $t-t_i\in d(\bar{\rho}^\prime)$.
Given $\bar{\rho}$ and $\bar{\rho}^\prime$, it can be inferred that $f_x(v)\in \mathcal{X}(\sigma_{o},t)$ according to Theorem 2.
%Suppose that the assumption (RO) does not hold such that the clock is not reset upon the occurrence of $e_{oi}$.
%That is to say, the clock value is $\theta$ after $e_{oi}$ is occurred.
%Given $v\in V_\lambda(x,[\theta,\theta],t-t_i)$, it can be inferred that $f_x(v)\in \mathcal{X}(\sigma_{o},t)$ according to Theorem 2 and Theorem 3.
%However, it is $V_\lambda(x,[\theta,\theta],t-t_i)\subseteq V_\lambda(x,z,t-t_i)$.
%Consider $v^\prime\in V_\lambda(x,z,t-t_i) \setminus V_\lambda(x,[\theta,\theta],t-t_i)$, it is $f_x(v^\prime)\notin \mathcal{X}(\sigma_{o},t)$ if there exists no $v\in V_\lambda(x,[\theta,\theta],t-t_i)$ such that $f_x(v)=f_x(v^\prime)$.
\end{proof}

This proposition shows that the state estimation can be updated by computing associated $\lambda$-estimations under Assumption RO.
In other words, given $(x,\theta)$ reached by a timed evolution in $\mathcal{X}(\sigma_o,t_i)$, $\mathcal{X}(\sigma_o,t)$ can be inferred by computing $V_\lambda(x,z,t-t_i)$, where $z\in Z(x)$ and $\theta\in z$.

Based on the previous results, Algorithm~\ref{alg:stateestimation} summarizes the proposed approach to compute $\mathcal{X}(\sigma_o,t)$.
%It can be explained as follows.
Consider a timed observation $(\sigma_o,t)$ with $\sigma_o=(e_{o1},t_1)\cdots(e_{on},t_n)\ (n\geq 1)$, where $e_{o1},\cdots, e_{on}\in E_o$.
The timed observation $(\sigma_o,t)$ is updated whenever an observable event $e_{oi}$ occurs at a time instant $t_i$, where $i\in\{1,\cdots,n\}$.
The algorithm provides the estimated states via a set of extended states $V_\lambda\subseteq V$ while time elapses in $[t_{i-1},t_i]$ with no event being observed, in addition to a set of extended states $\bar{V}_i\subseteq V$ of $G_z$ consistent with each new observation $(e_{oi},t_i)$, where $i\in\{1,\cdots,n\}$ and $t_0=0$.
Initially, it is imposed $\bar{V}_0=V_0$ and $\bar{V}_i=\emptyset$ for all $i\in\{1,\cdots,n\}$.
Then, for any $i\in\{1,\cdots,n\}$, the algorithm computes the $\lambda$-estimation from an extended state $(\bar{x},\bar{z})\in \bar{V}_{i-1}$ within $t_i-t_{i-1}$ implying the discrete states unobservably $(t_i-t_{i-1})$-reachable from $\bar{x}$ with a clock value in $\bar{z}$, and the set $\bar{V}_i$ is updated with the extended states reached by transitions labeled with $e_{oi}$ from the extended states in $V_{\lambda}$.
After the set $\bar{V}_n$ is determined, we initialize $V_\lambda$ to be empty and update $V_\lambda$ by including the $\lambda$-estimation for each $(\bar{x},\bar{z})\in \bar{V}_{n}$ within $t-t_{n}$.
Finally, we return the set of discrete states of $G$ associated with $V_\lambda$ as the set of discrete states consistent with $(\sigma_o,t)$.
\begin{algorithm}
\caption{State estimation of a TFA}
\label{alg:stateestimation}
%\LinesNumbered
\KwIn{A TFA $G$ with a set of initial discrete states $X_0$, a set of observable events $E_o\subseteq E$, a zone automaton $G_z=(V, E_\tau, \Delta_z, V_0)$, and a timed observation $(\sigma_o,t)$ from $0$ to $t\in\mathbb{R}_{\geq 0}$, where
$\sigma_o=(e_{o1},t_1)\cdots(e_{on},t_n)\ ( n\geq 1)$ and $t_1,\cdots,t_n\in \mathbb{R}_{\geq 0}$}
\KwOut{A set of discrete states $\mathcal{X} (\sigma_o, t)$}

let $\bar{V}_0=V_0$, $t_0=0$ and $X_{\lambda}=\emptyset$\

\For{each $i\in \{1,\cdots,n\}$ }
{
let $e=e_{oi}$, $\bar{V}_i=\emptyset$ and $V_\lambda=\emptyset$\
% realizes a $\lambda$-sensor\;

\For{each $(\bar{x},\bar{z})\in \bar{V}_{i-1}$}
{compute $V_\lambda(\bar{x},\bar{z}, t_i-t_{i-1})$\

let $V_\lambda=V_\lambda\cup V_\lambda(\bar{x},\bar{z}, t_i-t_{i-1})$\
%let $X_{\lambda i}=X_{\lambda i} \cup \{x\in X\mid (\exists z\in Z(x))(x,z)\in V_\lambda\}$\;
}

% realizes an event sensor\;
%\For{each $x\in X_{\lambda i}$}
%{
\For{each $v\in V_\lambda$}
{
let $x=f_x(v)$ and $z=f_z(v)$\

\If{$\exists x^\prime\in X$ s.t. $z\subseteq \Gamma((x,e,x^\prime))$}
{
\For{each $z^\prime\in Z(x^\prime)$ s.t. $z^\prime\subseteq Reset((x,e,x^\prime))$}
{
let $\bar{V}_i=\bar{V}_i\cup \{(x^\prime,z^\prime)\}$\
}
}
}

%}
}
let $V_\lambda= \emptyset$\

\For{each $(\bar{x},\bar{z})\in \bar{V}_{n}$}
{
compute $V_\lambda(\bar{x},\bar{z},t-t_n)$ \

let $V_\lambda=V_\lambda\cup V_\lambda(\bar{x},\bar{z}, t-t_{n})$\
}

\textbf{return} $\mathcal{X} (\sigma_o, t)=\{x\in X\mid (\exists z\in Z(x))(x,z)\in V_\lambda\}$
\end{algorithm}

The complexity of Algorithm~\ref{alg:stateestimation} depends on the size $n$ of the timed observation.
For each pair $(e_{oi},t_i)$, two {\em for} loops are executed: (1) the first {\em for} loop at Step 4 is executed at most $|V|$ times, computing $\lambda$-estimation whose complexity is $\mathcal{O}(q^3|X|)$, where $q=\max\limits_{x\in X} |Q_x|$ denotes the maximum number of zones for all discrete states, the complexity of this loop is $\mathcal{O}(q^4|X|^2)$; (2) the second {\em for} loop at Step 7 is executed at most $|V|$ times, and the {\em for} loop at Step 10 is executed at most $2q+1$ times; hence its complexity is $\mathcal{O}(q^2|X|)$.
%\begin{itemize}
%\item
%The first {\em for} loop at Step 4 is executed at most $|V|$ times, calling Algorithm~\ref{alg:lambdaestimation} whose complexity is $\mathcal{O}(q^3|X|)$. Thus the complexity of this loop is $\mathcal{O}(q^4|X|^2)$.
%\item
%The second {\em for} loop at Step 7 is executed at most $|V|$ times, and the {\em for} loop at Step 10 is executed at most $2q+1$ times; hence its complexity is $\mathcal{O}(q^2|X|)$.
%\end{itemize}
Finally, the {\em for} loop at Step 13, analogously to the  {\em for} loop at Step 4, has complexity $\mathcal{O}(q^4|X|^2)$.
Overall, the complexity of Algorithm~\ref{alg:stateestimation} is $\mathcal{O}(n(q^4|X|^2 +q^2|X|)+q^4|X|^2)= \mathcal{O}(nq^4|X|^2)$.

This approach allows one to construct an offline observer, i.e., a finite structure that describes the state estimation for all possible evolutions.
In more detail, during the online phase when the timed observation is updated by the latest measured observable event at a time instant, one can simply check to which interval (among a finite number of time intervals) the time elapsed between two consecutive updated observations belongs and to which discrete state the measured observable event leads.

\begin{table*}[t]
\caption{State estimation of the TFA $G$ in Fig. 1 with $\bar{X}_0=X_0$, $t_0=0$ and $(\sigma_o,t), t\in [0,4]$.}
\begin{center}\scalebox{0.55}{
\begin{tabular}{c|c|c|c|c|c}
\toprule[1pt]
$k$ & $\sigma_o$ &  \tabincell{c}{Time interval $I$\\ ($t\in I$)} & $V_\lambda=\bigcup\limits_{v\in \bar{V}_{k-1}} V_\lambda(f_x(v),f_z(v),t-t_{k-1}), t\in I$ & \tabincell{c}{$\mathcal{X}(\sigma_o,t)$}  & $\bar{V}_k$\\
\hline
\multirow{3}*{$1$}   &  \multirow{3}*{$\lambda$}  &  [0,0] &   $\{(x_0,[0,0]), (x_2,[0,0])$\}&  $\{x_0, x_2\}$  & \multirow{3}*{   \tabincell{c}{$\{(x_2,[0,0]),(x_4,[0,1])\}$}  }\\
&     &   (0,1) & $\{(x_0,(0,1)), (x_2,(0,1))\}$ &  $\{x_0, x_2\}$ & \\
&     &   [1,1] & $\{(x_0,[1,1]),(x_1,[1,1]),(x_2,[1,1]),(x_3,[1,1])\}$& $\{x_0, x_1, x_2, x_3\}$&\\
\hline
\multirow{6}*{$2$}   &  \multirow{6}*{$(a,1)$}  &  [1,1] &  $\{(x_2,[0,0]),(x_3,[0,0]),(x_4,[0,1])\}$ &  $\{x_2,x_3,x_4\}$  & \multirow{6}*{ $\{(x_2,[0,0])\}$ }\\
&     &   (1,2) & $\{(x_2,(0,1)), (x_3,[0,0]), (x_3,(0,1)), (x_4,[0,1])\}$ &  $\{x_2,x_3,x_4\}$ & \\
&     &   [2,2] & $\{(x_2,[1,1]), (x_3,[0,0]), (x_3,(0,1)), (x_3,[1,1]),(x_4,[0,1])\}$ &  $\{x_2,x_3,x_4\}$ & \\
&     &   (2,3) & $\{(x_2,(1,2)), (x_3,(0,1)), (x_3,[1,1]), (x_3,(1,2)),(x_4,(1,+\infty))\}$ &  $\{x_2,x_3,x_4\}$ &\\
&     &   [3,3] & $\{(x_2,[2,2]), (x_3,[1,1]), (x_3,(1,2)), (x_3,[2,2]),(x_4,(1,+\infty))\}$ &  $\{x_2,x_3,x_4\}$ & \\
\hline
\multirow{2}*{$3$}   &  \multirow{2}*{$(a,1)(a,3)$}  &  [3,3] & $\{(x_2,[0,0])\}$ & $\{x_2\}$  & \multirow{2}*{-}\\
&     &   (3,4) & $\{(x_2,(0,1))\}$& $\{x_2\}$ & \\
&     &   [4,4] & $\{(x_2,[1,1]),(x_3,[1,1])\}$& $\{x_2,x_3\}$ & \\
\bottomrule[1pt]
\end{tabular}}
\end{center}\label{estimation}
\end{table*}

\begin{example}\label{stateestimation}
Consider a timed observation $(\sigma_o,4)$, where $\sigma_o=(a,1)(a,3)$, produced by partially observed TFA $G$ in Fig.~\ref{tfaG} with $E_o=\{a\}$ and $E_{uo}=\{b,c\}$.
It implies that the observable event $a$ has been measured twice at $t_1=1$ and $t_2=3$, respectively, while the current time instant is $t=4$.
Table~\ref{estimation} shows how the state estimation is updated while time elapses in the time interval $[0,4]$ taking into account the two observations of event $a$.$\hfill\square$
%The state estimation is progressively updated over time as follows.
%
%\begin{itemize}
%  \item Step $k=1$: $G$ produces observation $(\sigma_o,t)=(\lambda,t)$ for $t\in [0,1]$. The state estimation process is progressively updated from $\bar{V}_0=\{(x_0,[0,0]),(x_2,[0,0])\}$.
%      According to $\bigcup\limits_{v\in \bar{V}_0}V_\lambda(f_x(v),f_z(v),1)=\{(x_0,[1,1]),(x_1,[1,1]),(x_2,[1,1]),(x_3,[1,1])\}$, it is $\mathcal{X}(\lambda,1) = \{x_0,x_1,x_2,x_3\}$.
%        After the pair $(a,1)$ is received, the state estimates update as $\bar{V}_1=\{(x_2,[0,0]),(x_4,[0,1])\}$ according to $((x_3,[1,1]),a,(x_2,[0,0]))\in\Delta_z$ and $((x_1,[1,1]),a,(x_4,[0,1]))\in\Delta_z$.
%  \item Step $k=2$: $G$ produces observation $(\sigma_o,t)=((a,1),t)$ for $t\in [1,3]$. According to $\bigcup\limits_{v\in \bar{V}_1}V_\lambda(f_x(v),f_z(v),2)=\{(x_2,[2,2])$, $(x_3,[1,1])$, $(x_3,(1,2))$, $(x_3,[2,2])$, $(x_4,(1,+\infty))\}$, it is $\mathcal{X}((a,1),3) = \{x_2, x_3, x_4\}$.
%      After the pair $(a,3)$ is received, the state estimates update as $\bar{V}_2=\{(x_2,[0,0])\}$ according to $((x_3,[1,1]),a,(x_2,[0,0]))\in\Delta_z$.
%
%  \item Step $k=3$: $G$ produces observation $(\sigma_o,t)=((a,1)(a,3),t), t\in [3,4]$.
%        According to $\bigcup\limits_{v\in \bar{V}_2}V_\lambda(f_x(v),f_z(v),1)=\{(x_2,[1,1]),(x_3,[1,1])\}$, it is $\mathcal{X}((a,1)(a,3), 4) = \{x_2, x_3\}$.$\hfill\square$
%\end{itemize}
\end{example}

\section{Conclusions and future work}
In this paper, we consider timed automata with a single clock.
%The transition relation, the timing function, and the clock resetting function specify the dynamics of a timed automaton.
Assuming that certain events are unobservable, %namely, their occurrences are silent,
we deal with the problem of estimating the current state of the system as a function of the measured timed observations.
%The proposed solution is based on partitioning the clock values at a discrete state into zones according to the firability of transitions.
By constructing a zone automaton that provides a purely discrete description of the considered TFA, the problem of investigating the reachability of a discrete state in the TFA is reduced to the reachability analysis of an extended state in the associated zone automaton.
%We construct a zone automaton that provides a purely discrete description of the considered TFA.
%Each state of the zone automaton is called an extended state, which is a pair of a discrete state of the associated timed automaton and a time interval, called a zone, specifying a range of the clock value.
%The problem of investigating the reachability of a discrete state in the TFA is reduced to the reachability analysis of an extended state in the associated zone automaton.
Assuming that the clock is reset upon each observable transition, we present a formal approach that can provide the set of discrete states consistent with a given timed observation and a range of the possible clock values.
%One can also obtain a range of the possible current value associated with the estimated discrete states by this approach.
%This approach allows one to construct an offline observer, i.e., a finite structure that describes the state estimation for all possible evolutions.
%An algorithm is proposed to summarize the approach that allows us to compute the set of discrete states consistent with a given timed observation at a certain time instant.
%The case of multiple clocks is not dealt with in detail in this paper, but a discussion on how to deal with it is proposed and illustrated with a numerical example.
%In this paper we have adopted a perspective of DESs and we have mainly considered the problem of estimating the discrete state of a timed automaton with a single clock. We believe that this approach can be generalized to the state estimation of more general hybrid systems (e.g., classes of rectangular automata with multiple clocks) where both the discrete state, i.e., the location, and the continuous state, i.e., the clock values, should be estimated based on the observed sequence of events.
%As a possible extension to multiple clocks, it is of great interest to investigate a structure, for instance, a product of zone automata of multiple clocks, that can provide the state estimation by checking its reachability.
Based on the timed automata model in this paper, it is worthy of investigating the state estimation approach regarding multiple clocks.
In addition, the proposed approach allows to construct an offline observer, which is fundamental to address problems of diagnosis, diagnosability, and feedback control, that are reserved for future works.

\end{document}